\documentclass[sigconf]{acmart}

\usepackage[ruled, vlined, linesnumbered]{algorithm2e}

\def\header{\vspace{1.5mm} \noindent}

\def\figcapup{\vspace{-1mm}}

\newtheorem{Lemma}{Lemma}

\def\api{\hat{\pi}}

\def\nbr{N}
\def\outnbr{N_{out}}
\def\innbr{N_{in}}
\def\few{{\text{Distributed-FAPPR}}}

\def\distppr{{\it DistPPR}}
\def\doubling{{\it Doubling}}
\def\gxppr{{\it GXPPR}}

\def\*{\star}

\def\done{\hspace*{\fill} $\framebox[2mm]{}$}	

\def\figcapup{\vspace{-2mm}}

\copyrightyear{2019}
\acmYear{2019} 
\setcopyright{iw3c2w3}
\acmConference[WWW '19]{Proceedings of the 2019 World Wide Web Conference}{May 13--17, 2019}{San Francisco, CA, USA}
\acmBooktitle{Proceedings of the 2019 World Wide Web Conference (WWW '19), May 13--17, 2019, San Francisco, CA, USA}
\acmPrice{}
\acmDOI{10.1145/3308558.3313555}
\acmISBN{978-1-4503-6674-8/19/05}

\title{Distributed Algorithms for Fully Personalized PageRank on Large Graphs}

\author{Wenqing Lin}
\affiliation{Interactive Entertainment Group, Tencent Inc.}
\email{edwlin@tencent.com}

\begin{document}
\begin{sloppy}

\begin{abstract}
	Personalized PageRank (PPR) has enormous applications, such as link prediction and recommendation systems for social networks, which often require the fully PPR to be known. Besides, most of real-life graphs are edge-weighted, e.g., the interaction between users on the Facebook network. However, it is computationally difficult to compute the fully PPR, especially on large graphs, not to mention that most existing approaches do not consider the weights of edges. In particular, the existing approach cannot handle graphs with billion edges on a moderate-size cluster. To address this problem, this paper presents a novel study on the computation of fully edge-weighted PPR on large graphs using the distributed computing framework. Specifically, we employ the Monte Carlo approximation that performs a large number of random walks from each node of the graph, and exploits the parallel pipeline framework to reduce the overall running time of the fully PPR. Based on that, we develop several optimization techniques which (i) alleviate the issue of large nodes that could explode the memory space, (ii) pre-compute short walks for small nodes that largely speedup the computation of random walks, and (iii) optimize the amount of random walks to compute in each pipeline that significantly reduces the overhead. With extensive experiments on a variety of real-life graph datasets, we demonstrate that our solution is several orders of magnitude faster than the state-of-the-arts, and meanwhile, largely outperforms the baseline algorithms in terms of accuracy.
\end{abstract}

\ccsdesc[500]{Information systems~Social networks}
\ccsdesc[300]{Information systems~Personalization}
\ccsdesc[300]{Computing methodologies~MapReduce algorithms}

\keywords{Distributed Algorithm; Graph; Personalized PageRank}

\maketitle

\section{Introduction} \label{sec:introduction}

Given a graph $G=(V,E,w)$, and two nodes $u, v \in V$, the personalized PageRank (PPR) of $v$ with respect to $u$, denoted by $\pi(u, v)$, is defined as the probability of a random walk starting from $u$ and ending at $v$. Compared to PageRank (PR), which is a global metric of nodes, PPR depicts the asymmetric perspective of each node to all the other nodes, i.e., $\pi(u,v)$ is not necessarily equal to $\pi(v,u)$. There are a plethora of applications that utilize PPR to measure the proximity of two nodes in the graph, such as recommendation system \cite{ppte15,cpjy+18,lxlk19} and natural language processing \cite{ea09}. In these applications, the fully PPR, i.e., the PPRs of all pairs of nodes, are often required. In other words, for all node $v$ in the graph, the PPR of all the other nodes $u$ with respect to $v$ should be known.

However, most of previous work \cite{qzjj16,txgj+17,jnlu17} focus on the computation of single-pair or single-source PPR, rendering them unsuitable for solving the problem of fully PPR. For example, FORA \cite{srxz+17} is the recent algorithm for single-source PPR, and needs 103 seconds to answer a single-source query on a Twitter graph of 41.7 millions of nodes, which would require about 137 years to compute the single-source PPRs for all nodes. In addition, it is highly difficult to compute the fully PPR, especially on large graphs. In particular, the probability of a random walk is usually computed by the Power Iteration method \cite{ymth+12} or approximated by the Monte Carlo method  \cite{dbkt05}, which might require a large number of iterations of computation to obtain a converged results. Besides, most of existing methods are tailored for the setting of a single machine \cite{dbkt05}, which renders them unable to handle graphs that cannot fit in memory, not to mention that the space complexity of the fully PPR could be $O(n^2)$ where $n$ is the number of nodes in the graph.
What's more, most real-world graphs are weighted, which can be directly generated from the interactions between nodes. For instance, the weight of the edge from a user $u$ to another user $v$ in the Facebook network could be the number of comments from $u$ to the posts of $v$. Nevertheless, most existing methods do not consider the weights on graph, making their algorithms unsuitable.

To address these issues, we devise a distributed share-nothing solution that employs the Monte Carlo approximation to perform a large number of random walks from all nodes of the graph in parallel. As such, for each node $u \in V$, we maintain only the PPR $\pi(u, v)$ where $v \in V$ is the node incurred in at least one random walk. Besides, we exploit the parallel pipeline framework \cite{wadp05} where (i) a subset of random walks is computed simultaneously in a pipeline and (ii) the new pipeline can remedy the deficiency of the old pipeline, as it is rare to generate long random walks due to the termination probability. Hence, it reduces the number of iterations of computation that significantly accelerates the process.

In order to optimize the above solution, we first address the issue of large nodes that would lead to the skewness of data. Specifically, we recursively divide each large node into small ones, which are then organized as the {\it alias tree} in a hierarchical manner. In addition, to address the issue of under-utilization of small nodes, we pre-compute for each small node a set of {\it short} walks to facilitate the acceleration of random walk generation on small nodes. Finally, since it would significantly degrade the performance of a pipeline which deals with very few or too many random walks at a time, we develop a method that optimizes the number of random walks in each pipeline so as to reduce the overhead of pipelining.

In summary, our contributions are the followings.

\begin{itemize}
  \item We devise the distributed algorithms based on the parallel pipeline framework that fully utilizes the parallelism to efficiently generate a large number of random walks simultaneously for all nodes of the graph.
 \item We design the hierarchical sampling algorithm that avoids the skewness on large nodes and also allows the algorithm to select a node from the set of neighbors following the distribution of weights on edges in constant time. 
  \item We develop an efficient algorithm to deal with small nodes by pre-computing short walks so that some random walks can be stopped earlier, which saves lots of costs. 
 \item We devise an effective method to optimize the number of random walks in each pipeline which largely increases the throughput of the algorithm that greatly improves its performance on large graphs. 
  \item In the extensive experiments, we demonstrate that our solution outperforms the state-of-the-art solutions by up to 21\% on accuracy and by up to 2 orders of magnitude on the running time.
\end{itemize}

\header{{\bf Paper organization}.} Section~\ref{sec:preliminaries} illustrates the definitions and notations used in the paper. An overview of our solution is depicted in Section~\ref{sec:overview}. After that, the details of the algorithm are explained in Section~\ref{sec:algorithms}. Section~\ref{sec:experiments} demonstrates the superior performance of our algorithm compared with several baseline methods over many graphs. Section~\ref{sec:related} discusses the related work of this paper. Finally, we conclude the paper in Section~\ref{sec:conclusions}.

\section{Preliminaries} \label{sec:preliminaries}

In this section, we present the concepts and frequently used notations in this paper, as listed in Table~\ref{tbl:def-notation}.

\subsection{Personalized PageRank} \label{sec:def-ppr}

Consider an edge-weighted graph $G=(V,E,w)$, where $V$ is the set of nodes, $E$ is the set of edges, and $w$ is the weighting function that maps each edge $e \in E$ to a positive number $w(e)$, i.e., $w(e) > 0$. For each node $s \in V$, we say that $t$ is an {\it out-neighbor} (resp. {\it in-neighbor}) of $s$, if there exists an edge $\langle s,t \rangle$ (resp. $\langle t,s \rangle$) such that $\langle s,t \rangle \in E$ (resp. $\langle t,s \rangle \in E$). Denote $\outnbr(s)$ (resp. $\innbr(s)$) as the set of out-neighbors (resp. in-neighbors) of $s$.

Given an out-neighbor $t \in \outnbr(s)$ of a node $s \in V$, we refer to the {\it routing probability} of $t$ with respect to $s$ as $r(s,t) = w(\langle s,t \rangle)/ \sum_{v \in \outnbr(s)} w(\langle s,v \rangle)$. As such, we have $r(s,t) \in (0, 1]$.

Given a graph $G=(V,E,w)$, a {\it walk} or {\it path} $P$ of $G$ is a sequence of nodes, denoted by $\langle v_1, v_2, \cdots, v_k \rangle$, such that (i) $v_i \in V$ for each $i \in [1, k]$ and (ii) $\langle v_i, v_{i+1} \rangle \in E$ for each $i \in [1, k-1]$. The {\it length} of $P$, denoted by $\ell(P)$, is the number of edges in $P$, i.e., $\ell(P)=k-1$. 

A {\it random walk} $P_r = \langle v_1, v_2, \cdots, v_k \rangle$ of $G$ is a path that is generated in a random manner following the distribution of routing probability. Specifically, for each node $v_i$ where $i \in [1, k-1]$, we have $v_{i+1} \in \outnbr(v_i)$ and $v_{i+1}$ is chosen with the routing probability $r(v_i, v_{i+1})$. Furthermore, at each step of the random walk $P_r$, there is a {\it termination probability} $\alpha \in (0, 1)$, which is a user-defined number, that determines whether the path will be terminated or not. In other words, after choosing $v_{i+1}$, $P_r$ will stop at $v_{i+1}$ with the termination probability $\alpha$. Hence, the expected length of a random walk is $1/\alpha$. In case that $\outnbr(v_i)$ is empty, following the previous work \cite{srxz+17}, we restart $P_r$ from its head, i.e., $v_1$.

\begin{example}
Figure~\ref{fig:graph} shows a graph $G$ which has 8 nodes and 12 weighted edges. The weights of edges incident to each node are normalized such that the routing probability of each edge equals its weight. Consider node $v_0$ and $\alpha = 0.5$. We have $\outnbr(v_0) = \{v_1, v_2, v_3, v_4, v_5, v_6\}$ and $\innbr(v_0) = \{v_5\}$. The path $\langle v_0, v_3, v_4 \rangle$ is walk starting from $v_0$ and ending at $v_4$, which can be generated by (i) first selecting $v_3$ from $\outnbr(v_0)$ with the probability $0.15$, (ii) and then with the probability of $(1-\alpha)=0.5$ continuing to select $v_4$ from $\outnbr(v_3)$ with the probability $1$, (iii) finally stopping at $v_4$ with the termination probability $\alpha=0.5$. Besides, there exist the other paths that start from $v_0$ and end at $v_4$, such as $\langle v_0, v_4 \rangle$ and $\langle v_0, v_2, v_3, v_4 \rangle$.
\done
\end{example}

Given a graph $G=(V,E,w)$ and two nodes $u,v \in V$, the {\it personalized PageRank} (PPR) of $v$ with respect to $u$, denoted by $\pi(u, v)$, is the probability of a random walk which starts from $u$ and ends at $v$.

To address this problem, there are roughly two kinds of solutions: One is the matrix based solution that utilizes the adjacency matrix of the graph and the matrix of fully PPR, the other one is based on the Monte Carlo (MC) method that estimates the probability by simulating a large number of random walks on the graph. However, the matrix-based solution incurs a space complexity of $O(n^2)$ where $n$ is the number of nodes in the graph, which is extremely huge, especially for large graph. As such, we adopt the MC based solution in our proposed approach.

Specifically, for each node $s \in V$, the MC method generates a number $\omega$ of random walks starting from $s$ to estimate $\pi(s, t)$ for every other node $t \in V$. If there are $\omega'$ random walks terminating at $t$, then we have an unbiased estimate of $\pi(s, t)$ as $\api(s, t) = \omega'/\omega$. It is easy to see that a larger $\omega$ leads to a more accurate estimate of $\pi(s, t)$, and also a much more expensive computational cost. To strike a good trade-off, we follow the theorem in the work \cite{dbkt05} by setting $\omega = \Omega(\frac{\log(1/p_f)}{\epsilon^2 \delta})$ to obtain an $\epsilon$-approximate PPR for any $\pi(s,t) \geq \delta$ with a probability $1-p_f$, where $s,t \in V$, which is formally defined as follows.

\begin{figure}[!t]
\centering
    \begin{tabular}{c}
    \includegraphics[height=30mm]{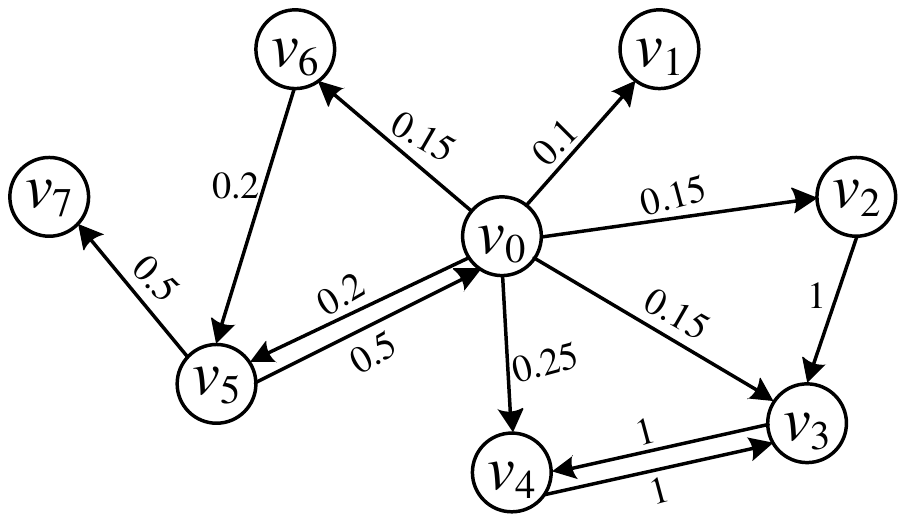}
    \end{tabular}
\vspace{-1mm}
\figcapup
    \caption{A toy graph with 8 nodes and 12 weighed edges.}  \label{fig:graph}
\end{figure}

\begin{definition}[Fully Approximate PPR (FAPPR)]\label{def:approx-ppr}
Given an edge-weighted graph $G=(V,E,w)$, a threshold $\delta$, an error bound $\epsilon$, and a failure probability $p_f$, the fully approximate PPR returns an estimated PPR $\api(s,t)$ for all pairs $(s,t)$ of nodes where $s,t \in V$, such that for all $\pi(s,t) \ge \delta$,
\begin{align}
|\pi(s,t) - \api(s,t)| \leq \epsilon \cdot \pi(s,t) \label{eq:result_quality}
\end{align}
holds with a probability at least $1 - p_f$. \done
\end{definition}

\header
{\bf Goal}. In the present work, we are to devise efficient distributed algorithms for the computation of fully approximate PPR on large graphs.

\subsection{The Alias Method} \label{sec:alias}

One of the key operations in random walk is to select one node $t$ from $s$'s out-neighbors, where $s \in V$, according to the routing probability. In other words, given a set $S$ of elements, as well as the routing probability $r(e)$ for each $e \in S$, where $\sum_{e \in S} r(e) = 1$, we are to select an element $e$ from $S$ with the probability $r(e)$.

In a naive approach, similar to \cite{pp06}, we can identify an element $e$ from $S$ which has the largest $z^{1/r(e)}$, where $z \in [0, 1]$ is a random number. Consequently, this approach needs to inspect each element in $S$ on the fly, whose time complexity is $O(|S|)$. However, it might lead to the workload imbalance in the distributed computing algorithm on the graph, as the sets of out-neighbors could be of different sizes.

We present {\it the alias method} \cite{m91}, which pre-computes a {\it switch probability} $p(e) \in [0, 1]$ and an {\it alias} $a(e) \in S$ for each element $e \in S$. As such, an element $e$ can be selected from $S$ with constant time complexity by exploiting $p(e)$ and $a(e)$.

Algorithm~\ref{alg:alias} illustrates the pre-computing step of the alias method that takes as input the routing probability $r$ and the set $S$ of elements, i.e., edges. Initially, for each edge $e \in S$, its alias $a(e)$ is set as $e$, and its switch probability $p(e)$ is computed as $|S| \cdot r(e)$ (Line 1). Thus, we have $\sum_{e \in S} p(e) = |S|$ at the beginning of the algorithm. Denote $S^l$ (resp. $S^s$) as the set of elements $e$ where $p(e) > 1$ (resp. $p(e) < 1$) (Line 2). Then, we choose an element $x$ from $S^s$ and an element $y$ from $S^l$ respectively. After that, we update the alias of $x$ as $a(x) = y$, and decrease the probability of $y$ by $(1-p(x))$, referred to as $p'(y) = p(y) - (1-p(x))$. As $x$ is processed, we remove $x$ from $S_s$. Besides, If $p'(y) \leq 1$, then we remove $y$ from $S^l$. Furthermore, if $p'(y) < 1$, then we add $y$ into $S^s$. We iteratively process the elements in $S^l$, until $S^l$ is empty (Line 3-9). Finally, we have $p(e)$ and $a(e)$ for all $e \in S$ as the results. Since each element in $S$ is inspected only once, the time complexity of this pre-computing step is $O(|S|)$.

\begin{example} \label{emp:alias-setup}
Consider the set $S_1 = \{e_1, e_2, e_3\}$ in Fig.~\ref{fig:hierarchical-alias}. Consider the edges $e$ in $S_1$. The routing probability of $e_1$ is computed as $r(e_1) = w(e_1) / \sum_{e \in S_1} w(e) = 0.25$. Similarly, we have $r(e_2) = r(e_3) = 0.375$. At the beginning, the switch probability of $e_1$ is computed as $p(e_1) = |S_1| \cdot r(e_1) / \sum_{e \in S_1} r(e) = 0.75$. Similarly, we have $p(e_2) = p(e_3) = 1.125$. Besides, we have $a(e_i) = e_i$ where $1 \leq i \leq 3$. As such, we construct $S_1^l = \{e_2, e_3\}$ since $p(e_2) > 1$ and $p(e_3) > 1$, and $S_1^s = \{e_1\}$ since $p(e_1) < 1$. Assume that $e_1 \in S_1^s$ and $e_2 \in S_1^l$ are selected. Then, we update $a(e_1)$ as $e_2$, remove $e_1$ from $S_1^s$, and decrease $p(e_2)$ by $(1 - p(e_1))$, resulting in $p(e_2) = 0.875$. Since $p(e_2) < 1$, we remove $e_2$ from $S_1^l$ and add $e_2$ into $S_1^s$. Then, we have $S_1^s = \{ e_2 \}$ and $S_1^l = \{ e_3 \}$. In the same way, we can compute that $a(e_2) = e_3$ and $p(e_3) = 1$, which forms the aliases and switch probabilities of $S_1$, as shown in Fig.~\ref{fig:hierarchical-alias}.
\done
\end{example}

To select an element $e$ from $S$ with the switch probability $p(e)$ and the alias $a(e)$, we first identify an element $e'$ with the probability $1/|S|$. Then, a probability threshold $p_t \in [0, 1]$ is generated by random. If $p_t \leq p(e')$, then we have $e = e'$; Otherwise, let $e$ be the alias of $e'$, i.e., $e = a(e')$. Therefore, the probability of selecting $e$ from $S$ is $\frac{1}{|S|} \cdot p(e) + \sum_{e' \in S \wedge a(e') = e}\frac{1}{|S|} \cdot (1-p(e')) = r(e)$. Apparently, the time complexity of selecting by the alias method is $O(1)$.

\begin{algorithm}[!t]
	\caption{Alias($S$, $r$)}\label{alg:alias}
	\KwIn{The set $S$ of elements and the routing probability $r$}
	\KwOut{The switch probability $p(e)$ and the alias $a(e)$ for all $e \in S$}

    Let $p(e) = |S| \cdot r(e)$ and $a(e) = e$\;
    Let $S^l = \{ e \in S | p(e) > 1 \}$ and $S^s = \{ e \in S | p(e) < 1 \}$\;
    \While{\textnormal{$S^l$ is not empty}}
    {
        Select any elements $x \in S^s$ and $y \in S^l$\;
        Let $a(x) = y$ and remove $x$ from $S^s$\;
        Decrease $p(y)$ by $(1 - p(x))$\;
        \If{\textnormal{$p(y) \le 1$}}
        {
            Remove $y$ from $S^l$\;
            If $p(y) < 1$, then add $y$ into $S^s$\;
        }
    }
    \Return $p(e)$ and $a(e)$ for all $e \in S$.
\end{algorithm}

\subsection{Distributed Computing} \label{sec:dist-graph}

However, it is impractical to compute FAPPR on a single machine, since (i) the graph could be massive that cannot fit in the memory of a single machine, and (ii) there could be an extremely large amount of computational cost due to lots of graph traversals, especially when the graph is sufficiently large, which also results in a large number of random walks (see Section~\ref{sec:def-ppr}).

As such, we utilize the distributed computing to generate the large number of random walks for all nodes of the graph among massive machines in parallel. In particular, we adopt Spark \cite{mmms+10}, whose fundamental computing unit is MapReduce \cite{js04,ywx13}. Note that, it is easy to extend to the other parallel computing framework, such as Pregel \cite{gmaj+10}. In a MapReduce job, there are mainly two consequent phases: The first one is the Map phase that takes as input the data and emits key-value pairs by the map function; After that, it is the Reduce phase that aggregates the key-value pairs by the key and then applies the reduce function on the aggregated data, which outputs the new key-value pairs.

To facilitate the processing of graphs on Spark, we store the graph as key-value pairs, such that, for each node $v \in V$, we a key-value pair $\langle v, \outnbr(v) \rangle$ where the key is $v$ and the value is the set of $v$'s out-neighbors, i.e., $\outnbr(v)$. Denote the set of key-value pairs of $G$ as $KV(G) = \{ \langle v, \outnbr(v) \rangle | v \in V \}$. Besides, to present a path $P = \langle v_h, \cdots, v_t \rangle$, we use a key-value pair $\langle v_t, v_h \rangle$ consisting of the tail $v_t$ of $P$ as the key and the head $v_h$ of $P$ as the value. This is because (i) only the head and tail of a random walk can contribute to the calculation of PPR, and (ii) the extension of a path happens at its tail.

Based on that, given a set $\mathcal{P}$ of paths, we can make a further {\it move} or {\it step} for all paths $P \in \mathcal{P}$ by one {\it round} of computation on Spark that utilizes a MapReduce job: First, the Map phase takes as input both the key-value pairs $\langle v_h, v_t \rangle$ of all paths $P \in \mathcal{P}$ and the key-value pairs $\langle v, \outnbr(v) \rangle \in KV(G)$ of the graph $G$, and emits the key-value pairs as what they are; Then, in the Reduce phase, we aggregates the key-value pairs by the key, resulting in that each aggregated data has one $\langle v_h, \outnbr(v_h) \rangle$ and several other key-value pairs $\langle v_h, u \rangle$ that represent the paths starting from $u$ and currently visiting $v_h$. For each $\langle v_h, u \rangle$, we select a node $v'_h$ from $\outnbr(v_h)$ with the routing probability $r(v_h, v'_h)$, and output a new key-value pair $\langle v'_h, u \rangle$, which represents the extended path of $\langle v_h, u \rangle$.

Note that, one can implement the above algorithm by exploiting the {\it join} operation on Spark. As such, it can be optimized by a careful designed data partitioning scheme \cite{wdyy16}, which can avoid the cost of transferring between machines the set $KV(G)$ of key-value pairs of $G$.

\begin{table}[t]
\small
\centering
\caption{Frequently used notations.}\label{tbl:def-notation}
\vspace{-2mm}
 \begin{tabular} {|l|p{2.5in}|} \hline
   Notation             &   Description
   \\ \hline
   {\scriptsize $G = (V,E,w)$}        &   an edge-weighted and directed graph, where $V$ is the set of nodes, $E$ is the set of edges, and $w$ is the weighting function that maps each edge in $E$ to a positive number. \\ \hline
   $r$        &   the routing probability of an edge (see Section~\ref{sec:def-ppr}). \\ \hline
   $\omega$   &   the number of random walks for each node (see Definition~\ref{def:approx-ppr}). \\ \hline
   $\gamma$   &   the number of random walks in each pipeline (see Definition~\ref{def:approx-ppr}). \\ \hline
   $\alpha$   &   the termination probability of a random walk (see Definition~\ref{def:approx-ppr}). \\ \hline
   $\epsilon$ &   the relative accuracy guarantee (see Definition~\ref{def:approx-ppr}). \\ \hline
   $\delta$   &   the PPR value threshold (see Definition~\ref{def:approx-ppr}). \\ \hline
   $p$        &   the switch probability of an edge in the alias tree. \\ \hline
   $a$        &   the alias of an edge in the alias tree. \\ \hline
   $d$        &   the degree threshold in Algorithm~\ref{alg:hierarchical-alias} and Algorithm~\ref{alg:big-move}. \\ \hline
\end{tabular}
\end{table}

\begin{figure*}[t]
\centering
    \begin{tabular}{c}
\includegraphics[height=36mm]{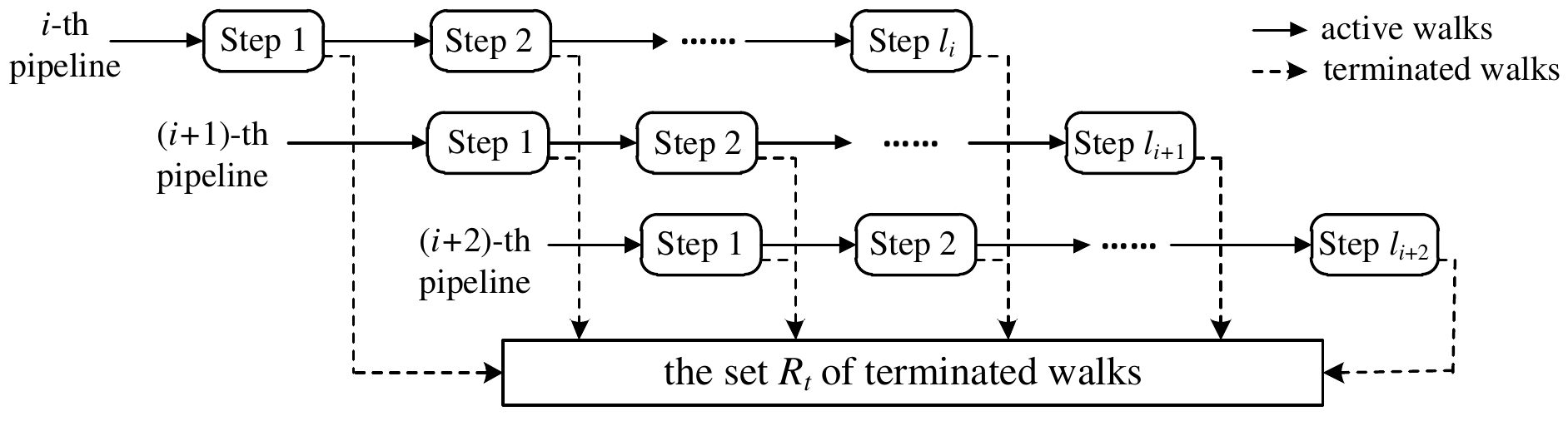}
    \end{tabular}
    \caption{This is an illustration of the parallel pipeline framework. In this framework, there are $\lceil \omega/\gamma \rceil$ pipelines, each of which generates at most $\gamma$ random walks by several steps. For each $i \in [1, \lceil \omega/\gamma \rceil)$, the $(i+1)$-th pipeline starts, only when the random walks in the $i$-th pipeline have completed the first step. All terminated walks in the pipelines are maintained in the set $R_t$, which are the final output of the computation of random walks.}
    \label{fig:parallel-pipeline}
\end{figure*}

\section{Solution Overview} \label{sec:overview}

\begin{algorithm}[t]
	\caption{$\few$($G$, $\epsilon$, $p_f$, $\delta$, $\alpha$)}\label{alg:overview}
	\KwIn{Graph $G = (V,E,w)$, relative accuracy guarantee $\epsilon$, failure probability $p_f$, threshold $\delta$, termination probability $\alpha$}
	\KwOut{The estimated PPR $\api(s,t)$, for all $s,t \in V$}

    Let $R_t = \emptyset$ be the collection of terminated random walks\;
    Let $R_a = \emptyset$ be the collection of active random walks\;
    Let $\omega = \Omega(\frac{\log(1/p_f)}{\epsilon^2 \delta})$ be the number of random walks to be computed for each node $v \in V$\;

    \While{\textnormal{$S \neq \emptyset$ or $\omega > 0$}}
    {
        \If {$\omega > 0$}
        {
            \For{\textnormal{each node $v \in V$ {\bf in parallel}}}
            {
                Add into $R_a$ $\min(\gamma, \omega)$ number of paths consisting of a single node $v$, i.e., $\langle v,v \rangle$\;
            }
            Update $\omega$ as $\omega - \gamma$\;
        }
        \For {\textnormal{each path $P \in R_a$ {\bf in parallel}}}
        {
            Let $s$ and $u$ be the head and tail of $P$ respectively\;
            Select a node $t \in \outnbr(u)$\;
            Append $t$ to the end of $P$\;
            Generate a random number $q \in [0, 1]$\;
            \If{$q \leq \alpha$}
            {
                Mark $P$ as terminated\;
            }
        }
        Remove from $R_a$ the terminated paths, which are then added into $R_t$\;
    }
    Let $R_t(s,t)$ be the set of paths $P \in R_t$ where $s$ and $t$ are the head and tail of $P$ respectively\;
    Compute $\api(s,t) = \sum_{P \in R_t(s,t)} 1/\omega$ in parallel\;
    \Return $\api(s,t)$ for all $s,t \in V$.
\end{algorithm}

In this section, we present an overview of our distributed algorithm for the computation of the fully approximate edge-weighted PPR (see Definition~\ref{def:approx-ppr}). As the major cost of FAPPR resides in the computation of random walks, in the sequel, we will concentrate our discussions on the algorithms for generating random walks.

Given a graph $G=(V,E,w)$, to generate $\omega$ random walks for each node $v \in V$ in the distributed computing framework, one straightforward solution could run in $\omega$ iterations, each of which computes one random walk for each node. In particular, one random walk $P$ ending with the node $v$ can be extended by one more step in one round of the distributed computing, where we append $P$ with one node randomly chosen from $\nbr(v)$. However, as the expected length of a random walk is $1/\alpha$, this approach results in $\omega / \alpha$ number of rounds of computation, which is highly expensive. Alternatively, one might propose to generate all $\omega$ random walks for each node $v \in V$ using one iteration with only $1/\alpha$ rounds, where $\omega$ number of paths for each node are computed simultaneously.
However, this approach would incur significant overhead in the management of memory, as $\omega$ could be sufficiently large.

To remedy the aforementioned issues, we exploit the parallel pipeline framework, as depicted in Figure~\ref{fig:parallel-pipeline}, where (i) each pipeline generates at most $\gamma \in [1, \omega]$ random walks for each node $v \in V$, and (ii) the pipelines begin sequentially, such that one pipeline $L$ starts only when the pipeline $L'$, prior to $L$, completes one round of computation, i.e., the paths in $L'$ are extended with one step. Therefore, there are $\lceil \omega / \gamma \rceil$ pipelines, and the expected number of rounds in each pipeline is $1/\alpha$. That is, the parallel pipeline framework is expected to employ $(\lceil \omega / \gamma \rceil + 1/\alpha - 1)$ number of rounds of computation, which is much smaller than the one of previous approaches. Besides, as each random walk could be terminated with the probability $\alpha$, the workload of each pipeline becomes fewer and fewer as the number of iterations increases. Therefore, we can avoid the under-utilization of computing for long pipelines, and make sure that the pipeline $L$ would receive few affections from the pipelines that are far from $L$.

Algorithm~\ref{alg:overview} depicts an overview of the proposed algorithm that takes as input the graph $G=(V,E,w)$, the relative accuracy guarantee $\epsilon$, failure probability $p_f$, threshold $\delta$, and termination probability $\alpha$. Initially, we have the empty collections of terminated and active random walks, referred to as $R_t$ and $R_a$ respectively (Lines 1-2). Besides, the number $\omega$ of random walks is calculated as $\Omega(\frac{\log(1/p_f)}{\epsilon^2 \delta})$. Then, we generate the random walks by several iterations (Lines 4-16). Each iteration consists of three phases, as follows. In the first phase, for each node $v \in V$, if $\omega > 0$, then we create $\min(\gamma, \omega)$ number of paths $\langle v,v \rangle$ consisting of a single node $v$, and decrease $\omega$ by $\gamma$ (Lines 5-8). In the second phase, for each path $P \in R_a$, we extend $P$ with one step by selecting a node $t$ the set of out-neighbors of the tail $u$ of $P$ according to the routing probability $r(u,t)$. After that, we append $t$ to the end of $P$ and mark $P$ as terminated if $\alpha$ is not smaller than a random number generated on the fly (Lines 12-15). Finally, we remove from $R_a$ the terminated paths, which are then added into $R_t$ (Line 16). At the end of the algorithm, to compute $\api(s,t)$ for all $s,t \in V$, we calculate the number of paths $P$ in $R_t$ such that the head and tail of $P$ are $s$ and $t$ respectively.

\section{Optimization Techniques} \label{sec:algorithms}

In this section, we develop several efficient techniques that significantly improve the performance of  Algorithm~\ref{alg:overview}.

\subsection{Handling Large Nodes: The Alias Tree} \label{sec:alias-tree}

As explained previously, we adopt the alias method which allows us to select one node $t$ from $s$'s out-neighbors, where $s \in V$, according to the routing probability with constant time and linear space complexity.
However, when the degree of a node $u \in V$ is sufficiently large, due to that the size of $\outnbr(u)$ could be too large to fit in memory. To avoid this deficiency, we devise a hierarchical approach that recursively divides $\outnbr(u)$ into several subsets, each of which has a size not bigger than $d$, that can fit in memory, where $d$ is a user-defined parameter and can be estimated according to the amount of available memory. Then, we denote each subset $S' \subseteq S$ as a new item with a weight equal to $\sum_{e \in S'} r(e)$. As such, we construct a new set $\mathcal{S}$ of items, whose alias and switch probability can be computed by the alias method on subsets of $S$. We build the {\it alias tree} by letting the alias and switch probability of $S'$ be a child of $\mathcal{S}$ for each item $S' \in \mathcal{S}$. Apparently, the alias tree is multi-way tree, where each node has a similar number of children. If the size of the set $\mathcal{S}$ of the subsets exceeds the size of memory, we recursively divide $\mathcal{S}$ in the same way.

\begin{algorithm}[t]
	\caption{Construct-Alias-Tree($S$, $r$, $d$)}\label{alg:hierarchical-alias}

	\KwIn{The set $S$ of elements, the routing probability $r$, and the degree threshold $d$}
	\KwOut{The alias tree $T(S)$ of $S$}

    \If{$|S| \leq d$}
    {
        Let $T(S)$ be the output of the alias method on $S$\;
    }
    \Else
    {
        Let $b = \lceil |S| / d \rceil$\;
        Split $S$ into $b$ disjoint subsets by random, referred to as $\mathcal{S} = \{ S_1, S_2, \cdots, S_b \}$\;
        Let $w'( S_i ) = \sum_{e \in S_i} r(e)$ for each $i \in [1, b]$\;
        Let $r'( S_i ) = \frac{w'(S_i)}{\sum_{S' \in \mathcal{S}} w'(S')}$ for each $i \in [1, b]$\;
        Let $r_i( e ) = \frac{r(e)}{\sum_{e' \in S_i} r(e')}$ for each $e \in S_i$ and $i \in [1, b]$\;
        For each $i \in [1, b]$, compute $f(S_i)$ by the alias method, which is treated as a child of $\mathcal{S}$\;
        Build $T(S)$ with Construct-Alias-Tree($\mathcal{S}$, $r'$, $d$)\;
    }
    \Return $T(S)$.
\end{algorithm}

Algorithm~\ref{alg:hierarchical-alias} provides the pseudo-code of the construction of an alias tree on a given set $S$ of items. The algorithm first inspects the size of $S$. If its size is not bigger than $d$, then the algorithm returns the alias directly computed on $S$ by the alias method, which forms a single node of the alias tree (Line 2). Otherwise, we randomly split $S$ into $b$ disjoint subsets (Line 5). Denote the set of subsets by $\mathcal{S} = \{S_1, S_2, \cdots, S_b \}$. Hence, we have $\cup_{S' \in \mathcal{S}} S' = S$ and $\cap_{S' \in \mathcal{S}} S' = \emptyset$. Then, we compute the weight of $S' \in \mathcal{S}$, which is the sum of routing probability of items in $S'$, referred to as $w'(S') = \sum_{e \in S'} r(e)$. Then, for each subset $S_i$ where $1 \leq i \leq b$, the routing probability of $S_i$ with respect to $\mathcal{S}$ is computed as $r'( S_i ) = \frac{w'(S_i)}{\sum_{S' \in \mathcal{S}} w'(S')}$ (Line 7). Besides, to compute the alias of each item $e$ in each subset $S_i$, the routing probability of $e$ with respect to $S_i$ is updated as $r'( S_i ) = \frac{w'(S_i)}{\sum_{S' \in \mathcal{S}} w'(S')}$ (Line 8). As a result, the child $f(S_i)$ of $\mathcal{S}$ is created by performing the alias method on the updated routing probability of items in $S_i$ (Line 9). Finally, we recursively construct the alias tree on $\mathcal{S}$, until the size of $\mathcal{S}$ fits in memory, i.e., not bigger than $d$ (Line 10).

\begin{example} \label{emp:graph}
Consider the 6 out-going edges of $v_0$ in Figure~\ref{fig:graph}, denoted by $e_1, e_2, \cdots, e_6$ respectively. As shown, in Figure~\ref{fig:hierarchical-alias}, the alias tree stores the 6 edges by treating them as 6 items in a set $S$. Assume that the degree threshold is $d = 3$. As the size of $S$ is larger than $d$, we randomly divide $S$ into two sets, e.g., $S_1 = \{e_1, e_2, e_3\}$ and $S_1 = \{e_4, e_5, e_6\}$. Hence, we have $w'(S_1) = r(e_1) + r(e_2) + r(e_3) = 0.4$. Similarly, we have $w'(S_2) = 0.6$.
As such, to obtain the routing probabilities, we can normalize the weights of elements in $S_1$, $S_2$, and $\mathcal{S}$ respectively. Then, by utilizing the alias method, we construct the aliases and switch probabilities of $S_1$ and $S_2$, which form the children of $\mathcal{S}$. After that, we recursively construct the alias of $\mathcal{S}$, which is returned as shown in Figure~\ref{fig:hierarchical-alias}, since the size of $\mathcal{S}$ is less than $3$.
\done
\end{example}

To select an item, we traverse the alias tree $T(S)$ recursively by starting from the root $T$ of $T(S)$. To explain, we first select an item $e$ from $T$ with a probability $1/|T|$, which can be achieved in $O(1)$ time. Then, a random number $r \in [0, 1]$ is generated. If $r \leq p(e)$, then we select $e$, otherwise $a(e)$. Afterwards, if $f(e)$ is not empty, i.e., $e$ has children, we recursively inspect the items in $f(e)$; otherwise, we return $e$ as the result.
Apparently, The running time of the selection process is determined by the height of the alias tree, i.e., $O(\log(|S|)/\log(d))$, where $d$ is the degree threshold.
Lemma~\ref{lemma:generate-sample} provides the correctness of random item generation based on the alias tree.

\begin{example} \label{emp:alias-select}
Consider the alias tree in Example~\ref{emp:alias-setup}. To select an element, it consists of two steps: First, we select an element $S'$ from $\mathcal{S}$; Then, an element $e$ is selected from $S'$. In the first step, assume that $S_1$ is selected with the probability $1/|\mathcal{S}| = 0.5$. After that, we generate a random number $r \in [0, 1]$. If $r \leq p(S_1)$, we select $S_1$, otherwise select $a(S_1) = S_2$. As a result, the probability of selecting $S_1$ is $0.5 \cdot p(S_1) = 0.4$, which equals the normalized weight of $S_1$, i.e., the routing probability of $S_1$. Assume that $S_1$ is selected in the first step. We recursively select an element from $S_1$ in the same way, which leads to the final selection.
\done
\end{example}

\begin{figure}[!t]
\centering
    \begin{tabular}{c}
    \includegraphics[height=35mm]{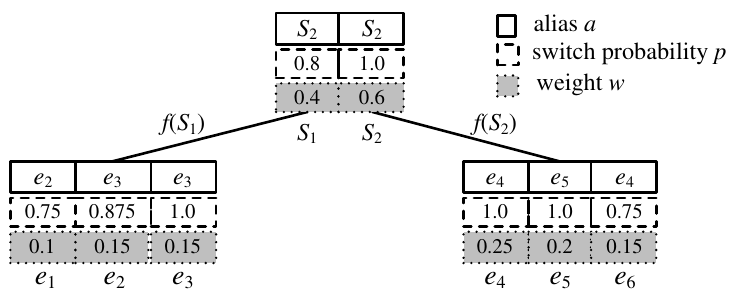}
    \end{tabular}
\vspace{-2mm}
    \caption{An example of the alias tree.}  \label{fig:hierarchical-alias}
\end{figure}

\begin{Lemma} \label{lemma:generate-sample}
Given a set $S$ of edges, as well as the alias tree $T$ constructed on $S$, the selection an element $e \in S$ from $T$ has the probability equal to $r(e)$.
\end{Lemma}

\begin{proof}
Let $\langle S_1, S_2, \cdots, S_i \rangle$ be the path from the root of $T$ to a leaf of $T$, where $e$ resides.
Consider a node $S_j$, where $1 \leq j \leq i$. As shown previously, each item $x$ in $S_j$ is selected with the probability $w(x) / \sum_{y \in S_j} w(y)$. Besides, the weight of $S_j$ is the sum of weights of items in $S_j$, i.e., $w(S_j) = \sum_{y \in S_j} w(y)$. Note that, $w(S_1) = \sum_{e \in S} r(e) = 1$, since $S_1$ represents all items in $S$. Therefore, the probability of selecting $e$ is
\[
\frac{w(S_2)}{w(S_1)} \cdot \frac{w(S_3)}{w(S_2)} \cdot \cdots \cdot \frac{w(e)}{w(S_i)}
= \frac{r(e)}{w(S_1)} = r(e).
\]
\end{proof}

\subsection{Handling Small Nodes: Make Big Moves} \label{sec:big-move}

Most real-life graphs follow the power law \cite{mpc99}, whose implication is that the majority of nodes in the graph have a very small set of neighbors (see Table~\ref{tbl:exp-real-data}). To speed up the random walks among the small-degree nodes, we devise a pre-computing method that generates the {\it big moves} for each small-degree node, such that each big move represents a random walk of several steps.
To facilitate the computation with big moves, we maintain each big move $P^b$ with the similar structure as the normal random walk with a difference in that we add a mark, denoted by $M(P^b)$, to represent whether $P^b$ is terminated or not. If $M(P^b) = 0$, then $P^b$ is a terminated path; otherwise, if $M(P^b) = 1$, then $P^b$ is an active path. As a result, we can denote a big move of $v$ by $\langle u,t_m,p \rangle $ which are the walks ending at $u$ with the termination mark $t_m$ and the probability $p$. In other words, compared with the set of out-neighbors, the big moves consists of the event of termination, which expands the probability space.

\begin{algorithm}[!t]
	\caption{Pre-compute-Big-Moves($G$, $V_s$, $d$, $\alpha$)}\label{alg:big-move}

	\KwIn{The graph $G=(V,E,w)$, the set $V_s \subseteq V$ of small nodes, the degree threshold $d$, and the termination probability $\alpha$}
	\KwOut{The big moves for all $v \in V_s$}

    Let $prevBSize = 0$\;
    Let $F(v) = \{ \langle u, p(v,u) \rangle \mid u \in \outnbr(v) \}$ for each $v \in V_s$\;
    Let $B(v) = F(v)$ for each $v \in V_s$\;

    \While{ $|B(v)| + |F(v)| < d$ \textnormal{and} $prevBSize < |B(v)|$ }
    {
        Update $prevBSize$ as $prevBSize = |B(v)|$\;
        Let $F'(v)$ be an empty collection for each $v \in V_s$\;
        \For{\textnormal{each node $\langle u, p(v,u) \rangle \in F(v)$ {\bf in parallel}}}
        {
            \For{\textnormal{each node $t \in \outnbr(u)$} {\bf in parallel}}
            {
                Add $\langle t, p(v,u)  \cdot (1-\alpha) \cdot p(u, t) \rangle$ into $F'(v)$\;
            }
        }
        Update $B(v)$ as the aggregation result on the pairs $\langle t, p \rangle$ in $F'(v) \cup B(v)$ by $t$\;
        Replace $F(v)$ with $F'(v)$\;
    }

    Let the set of terminated big moves be $P^b_t(v) = \{ \langle t, 0, p \cdot \alpha \rangle \mid \langle t, p \rangle \in B(v) \} $ for $v \in V_s$\;
    Let the set of active big moves be $P^b_a(v) = \{ \langle t, 1, p \cdot (1-\alpha) \rangle \mid \langle t, p \rangle \in F(v) \} $ for $v \in V_s$\;

    \Return $P^b_t(v) \cup P^b_a(v)$ for all $v \in V_s$.
\end{algorithm}

Algorithm~\ref{alg:big-move} describes the generation of big moves for each node $v$ in the set $V_s$ of small nodes. Initially, for each $v \in V_s$, we have the set $F(v)$ of active walks be the set of out-neighbors $u$ of $v$ with the accumulated probability equal to $p(v,u)$, i.e., $F(v) = \{\langle u, p(v, u) \rangle \mid u \in \outnbr(v)\}$. Besides, we let the set $B(v)$ of the terminated walks be initialized as $F(v)$, and we denote the size of $B(v)$ before updating (Line 10) by $prevBSize$, which is initialized as $0$. Note that, the initial sizes of $F(v)$ and $B(v)$ equal the size of $\outnbr(v)$. In what follows, the algorithm runs in several iterations. In each iteration, we extend each walk in $F(v)$ with one more step by adding each node $u \in \outnbr(v)$, resulting in a set $F'(v)$ of new walks. Then, we compute the accumulated paths in both $B(v)$ and $F'(v)$, and also update $F(v)$ with new walks (Lines 5-11). These processes are repeated in iterations, until (i) the number of terminated and active walks exceeds $d$, i.e., $|B(v)| + |F(v)| \geq d$, or (ii) the walks can not reach any new nodes, i.e., the size of $B(v)$ does not change any more, namely $prevBSize = |B(v)|$.  Finally, we integrate the walks in $B(v)$ and $F(v)$ by multiplying the probability in $B(v)$ and $F(v)$ with $\alpha$ and $1-\alpha$ respectively, and adding the termination mark 0 and 1 respectively, which leads to resulting set $BM(v)$ of big moves of $v$.

\begin{example} \label{emp:big-move}
Consider the graph $G = (V, E, w)$ in Figure~\ref{fig:graph}. Assume that the degree threshold $d=4$ and the termination probability $\alpha=0.5$. Then, all the nodes $v \in V$, except $v_0$, need to compute their big moves, since $|\outnbr(v)| < d$. Consider node $v_2$, which has one out-neighbor $v_3$. To compute the big move of $v_2$, we first initialize the set $F(v_2)$ as $\{ \langle v_3, 1 \rangle \}$ and the set $B(v_2) = F(v_2)$. Since $|F(v_2)| + |B(v_2)| = 2 < d$, we extend each walk in $F(v_2)$ by one step, and obtain the set $F'(v_2) = \{ \langle v_4, 1 \times 0.5 \times 1 = 0.5 \rangle \}$ as $v_3$ has only one neighbor, i.e., $v_4$. Then, we update $prevBSize=|B(v_2)|=1$, and update $B(v_2)$ as the aggregation of items in $B(v_2)$ and $F'(v_2)$, resulting in the set $\{ \langle v_3, 1 \rangle,  \langle v_4, 0.5 \rangle \}$. In addition, $F(v_2)$ is updated as $F'(v_2)$, i.e., $F(v_2) = \{\langle v_4, 0.5 \rangle \}$. Afterwards, we continue to extend the walks in $F(v_2)$, since $|F(v_2)| + |B(v_2)| = 3 < d$, and we have $F'(v_2) = \{ \langle v_3, 0.5 \times 0.5 \times 1 = 0.25 \rangle \}$. Now, we update $prevBSize = |B(v_2)| = 2$, update $F(v_2) = F'(v_2)$, and update $B(v_2)$ as the aggregation of items in $F'(v_2)$ and $B(v_2)$, leading to the set $\{ \langle v_3, 1.25 \rangle, \langle v_4, 0.5 \rangle  \}$. Because $prevBSize = |B(v_2)|$, i.e., the size of $B(v_2)$ does not change any more, we stop the iteration, and compute the set $P^b_t(v_2)$ of terminated big moves as $\{ \langle v_3, 0, 1.25 \times 0.5 = 0.625 \rangle, \langle v_4, 0, 0.5 \times 0.5 = 0.25 \rangle \}$ and the set $P^b_a(v_2)$ of active big moves as $\{ \langle v_3, 1, 0.25 \times 0.5 = 0.125 \rangle \}$. As a result, we have the big moves of $v_2$ as the set $BM(v_2) = P^b_t(v_2) \cup P^b_a(v_2) = \{ \langle v_3, 0, 0.625 \rangle, \langle v_4, 0, 0.25 \rangle, \langle v_3, 1, 0.125 \rangle \}$.
\done
\end{example}

We show the correctness of Algorithm~\ref{alg:big-move} by Lemma~\ref{lemma:big-move}.

\begin{Lemma} \label{lemma:big-move}
Given a set $BM(v)$ of big moves of $v$, then $BM(v)$ consists of all random walks starting from $v$ within $k$ steps, where $k$ is the longest length of walks in $BM(v)$.
\end{Lemma}

\begin{proof}
It is sufficient to show that (i) if there exists a path of length at most $k$ that starts from $v$ and ends at $u$, then $u$ must be in $BM(v)$; and (ii) $\sum_{\langle v, t_m, p \rangle \in BM(v)} p = 1$.

In the first case, we are to show that all paths, whose lengths are not larger than $k$, are in $BM(v)$. Let $P(v, u_t) = \langle v, u_1, u_2, \cdots, u_t \rangle$ be a path of length $t \leq k$ that starting from $v$. In the first round, we have $u_1$ selected, as $u_1$ is an out-neighbor of $v$, i.e., $u_1 \in \outnbr(v)$. By induction, $u_{i+1}$ is an out-neighbor of $u_i$, where $1 \leq i < t$. As such, the sequence $\langle v, u_1, u_2, \cdots, u_t \rangle$ is a path of $G$.

In the second case, we need to demonstrate that the probability space of $BM(v)$ is completed. Let $p_t$ be the probability that a random walk $P = \langle v, u_1, u_2, \cdots u_t \rangle$ is of length $t$. Hence, we have
\begin{align*}
p_t &= \sum_{\ell(P)=t} p( \langle v,u_1 \rangle ) p( \langle u_1,u_2 \rangle ) \cdots  p( \langle u_{t-1},u_t \rangle ) (1-\alpha)^{t-1} \\
&= (1-\alpha)^{t-1}.
\end{align*}
Consider the probability $p'$ of the event that a random walk $P \in BM(v)$ is a terminated path, which means that $P$ is of length $t \leq k$. As $BM(v)$ contains all paths of length not bigger than $k$, we have
\[
p' = \alpha \cdot \sum_{1 \leq t \leq k} p_t = \alpha \cdot \sum_{1 \leq t \leq k} (1-\alpha)^{t-1} = 1 - (1-\alpha)^k.
\]
Besides, consider the probability $p''$ of the event that a random walk $P \in BM(v)$ is an active path, which indicates that $P$ is of length $k$. Therefore, we have $p'' = (1-\alpha) \cdot p_k = (1-\alpha)^k$.
To sum up, we obtain $\sum_{\langle v, t_m, p \rangle \in BM(v)} p = p' + p'' = 1$, which completes the proof.
\end{proof}

To continue the generation of the random walk ending with the node $v$, which is a small node with the set $BM(v)$ of big moves, we randomly select a path $P^b = \langle v, t_m, p \rangle \in BM(v)$ following the distribution of probabilities of paths in $BM(v)$. Note that, the selection can be accomplished in constant time with the alias method. If $t_m = 0$, then $P^b$ is a terminated path, and we add $P^b$ into the set $R_t$. Otherwise, we continue the walk along the path $P^b$ until the termination.

\subsection{Optimizing \large{$\gamma$}} \label{sec:tune-gamma}

Our previous discussion has focused on the generation of random walks in the graph $G$. When the number $\omega$ of random walks is considerably large, it could be highly inefficient to generate all random walks for all nodes of $G$ at a time, as discussed in the previous sections. To remedy this deficiency, we proposed the parallel pipeline framework that processes the $\omega$ number of random walks by several iterations, each of which generates at most $\gamma$ random walks for each node of $G$. That is, there will be $\lceil \omega / \gamma \rceil$ number of pipelines.

One crucial question remains: how do we decide the number $\gamma$ of random walks in each pipeline? A straightforward approach is to set $\gamma=1$, i.e., we generate one random walk for each node $v \in V$ in each pipeline. Although this can avoid the memory issues as aforementioned, it brings up the other deficiencies where (i) the probability to having a long walk is low, which leads to the under-utilization of computing resources, and (ii) the number of pipelines would be a lot, which increases the running time. To tackle these issues, we devise a heuristic method to choose $\gamma$, as follows. First, in our implementation, we maintain a path $P$ by keeping the head $s$ and the tail $t$ of $P$, i.e., $\langle s,t \rangle$. As such, it is a constant space cost $c$ for a path $P$. Second, in the parallel pipeline framework, we start the pipelines sequentially such that one latter pipeline $L$ is started only upon that the pipeline $L'$, prior to $L$, completes the extension of random walks in $L'$ by one step. Specifically, let $s_i$ be the total number of active random walks that are generated by the first $i$ pipelines. After the extension of one step at the $i$-th round, the number of remaining random walks is expected to be $(1-\alpha) \cdot s_i$. In the $(i+1)$-th pipeline, we start another $n \cdot \gamma$ number of new random walks in total, where $n$ is the number of nodes of $G$. To sum up, the total expected number of random walks to be processed at the $(i+1)$-th round is $s_{i+1} = n \cdot \gamma + (1-\alpha) \cdot s_i$. Note that, at the beginning, we have $s_1 = n \cdot \gamma$. By solving these formulas, we have
\[
s_i = n \cdot \gamma \cdot \sum_{1 \leq j \leq i} (1-\alpha)^{j-1} = n \cdot \gamma \cdot \frac{1 - (1-\alpha)^i}{\alpha} < \frac{n \cdot \gamma}{\alpha}.
\]

As such, if the size of available memory on a machine is $M$, then we have $c \cdot n \cdot \gamma / \alpha \leq M$. That is, we obtain $\gamma \leq \alpha \cdot M / (n  \cdot c)$. Therefore, we can set $\gamma$ to the maximum integer such that the total size of processed random walks does not exceed the size of available memory, which helps the full utilization of the cluster. By experiments as shown later, we demonstrate that our choice of $\gamma$ results in good performance of the proposed algorithm.

\begin{table}[t]
\centering
\small
\caption{Datasets.}\label{tbl:exp-real-data}
\vspace{-2mm}
\begin{tabular}{| @{\;}c@{\;} | @{\;}r@{\;} | @{\;}r@{\;} | @{\;}r@{\;} | @{\;}r@{\;} | @{\;}r@{\;} | @{\;}r@{\;} |} \hline
Name& $|V|$ & $|E|$ & $d_{avg}$ & $d_{max}$ & $|V_{s}|/|V|$ & $|V_{l}|/|V|$ \\ \hline
{\it GrQc}	&5,241	&28,968	&5.5	&81	&78.3\%	&15.9\% \\ \hline
{\it CondMat}	&23,133	&186,878	&8.1	&279	&75.4\%	&12.2\% \\ \hline
{\it Enron}	&36,692	&367,662	&10.0	&1383	&85.5\%	&5.1\% \\ \hline
{\it DBLP}	&317,080	&1,049,866	&5.6	&306	&46.9\%	&4.2\% \\ \hline
{\it Stanford}	&281,903	&2,312,497	&8.2	&255	&75.5\%	&14.7\% \\ \hline
{\it Google}	&875,713	&5,105,039	&6.9	&456	&53.6\%	&1.9\% \\ \hline
{\it Skitter}	&1,696,415	&11,095,298	&11.5	&35,387	&51.0\%	&0.3\% \\ \hline
{\it Patents}	&3,774,768	&16,518,947	&7.9	&770	&38.5\%	&1.5\% \\ \hline
{\it Pokec}	&1,632,803	&30,622,564	&21.4	&8,763	&61.0\%	&2.6\% \\ \hline
{\it LiveJournal}	&4,846,609	&68,475,391	&15.9	&20,292	&66.5\%	&1.0\% \\ \hline
{\it Orkut}	&3,072,441	&117,185,083	&43.0	&33,007	&65.7\%	&3.3\% \\ \hline
{\it Twitter}	&41,652,230	&1,468,364,884	&36.6	&2,997,469	&87.9\%	&0.2\% \\ \hline
{\it Friendster}	&68,349,466	&2,586,147,869	&46.1	&5,214	&65.5\%	&12.8\% \\ \hline
\end{tabular}

\end{table}

\section{Experiments} \label{sec:experiments}

This section presents the thorough experimental studies on the performance of the proposed algorithm, compared with the state-of-the-art approach, on several real-life datasets.

\subsection{Experimental Settings}

We implement our distributed algorithm for the computation of FAPPR (dubbed as \distppr) in Scala on Spark 2.0\footnote{https://spark.apache.org/}, and compare it against two state-of-the-art distributed algorithms: {\doubling} \cite{bkd11} and the PPR library in GraphX\footnote{https://spark.apache.org/graphx/} (denoted as \gxppr). Note that, (i) {\doubling} is a MapReduce algorithm, which can be easily implemented on Spark, and (ii) {\gxppr} is a matrix-based method that iteratively computes with the adjacency matrix of the graph until the convergence of PPR.
All of our experiments are conducted on an in-house cluster consisting of up to $201$ machines, each of which runs CentOS, and has 16GB memory and 12 Intel Xeon Processor E5-2670 CPU cores. By default, we exploit $51$ machines to run the algorithms, where we have one machine as the master and the others as the workers in Spark.

We use $13$ real-life networks in our experiments, as shown in Table~\ref{tbl:exp-real-data}. These data are obtained from the SNAP collection\footnote{http://snap.stanford.edu/data/index.html}, and are intensively evaluated in the literatures \cite{srxz+17,zxxs+18}. Besides, these data are from various categories: {\it GrQc} and {\it CondMat} are collaboration networks; {\it Enron} is a communication network; {\it Standford} and {\it Google} are the web graphs; {\it Skitter} is a internet topology network; {\it Patents} and {\it DBLP} are the citation networks; {\it Pokec}, {\it LiveJournal}, {\it Orkut}, {\it Twitter}, and {\it Friendster} are social networks. Furthermore, these data have different sizes, ranged from thousands of edges to billions of edges. Their average out-degrees (denoted by $d_{avg}$) are relatively small, though their maximum out-degrees (denoted by $d_{max}$) could be significantly large that are more than 2 million. Let $V_s$ be the set of {\it small} nodes whose out-degrees are less than $d_{avg}$, and let $V_l$ be the set of {\it large} nodes whose out-degrees are more than $\sqrt{d_{max}}$. As shown in Table~\ref{tbl:exp-real-data}, the small nodes are the majorities in the graph, while the large nodes take only a small portion. Following previous work \cite{wdaj15}, we adopt the linear parameterization to assign each edge a positive weight for each graph.

In addition, according to previous work \cite{srxz+17,dbkt05}, we set the failure probability $p_f$ to $1/n$ where $n$ is the number of nodes in the graph, and set the default value of the relative accuracy guarantee $\epsilon$, the threshold $\delta$, and the termination probability $\alpha$ to $0.5$. Besides, if the algorithm does not terminate within 24 hours, we omit it from the experiments. We repeat each experiment 3 times and report the average reading of each approach.

\begin{figure*}[t!]
\centering
\begin{tabular}{c}
 \includegraphics[height=35mm]{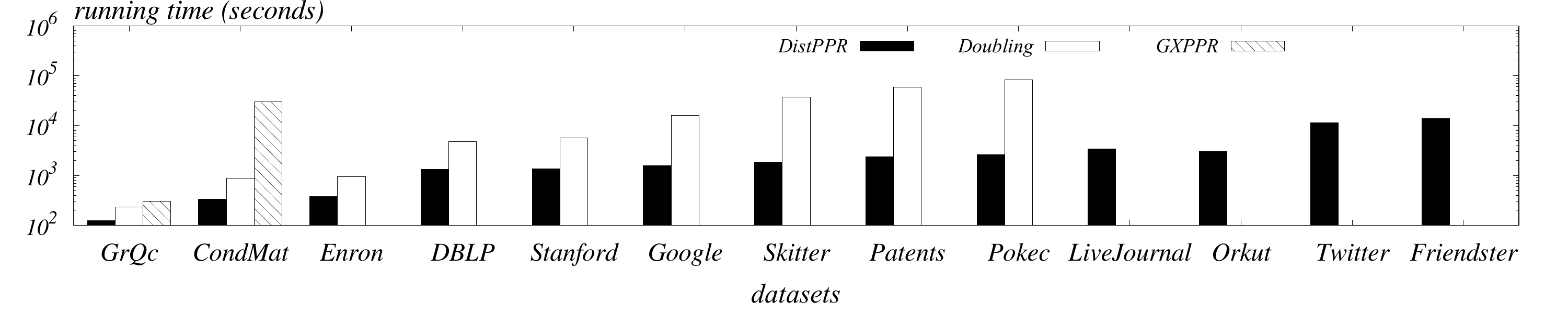} 
\end{tabular}
\vspace{-2mm}
\caption{Running time on real-life datasets.}
\label{fig:exp-efficiency-time}
\end{figure*}

\begin{figure*}[t!]
\centering
\begin{tabular}{c}
\hspace{-6mm} \includegraphics[height=36mm]{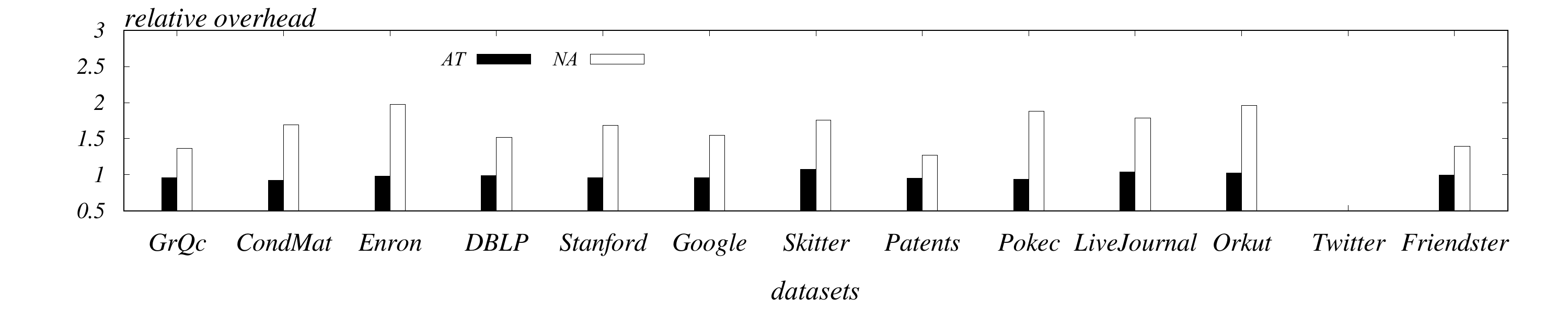} 
\end{tabular}
\vspace{-2mm}
\caption{Effects of optimizations.}
\label{fig:exp-efficiency}
\end{figure*}

\subsection{Comparisons with Previous Work}

In the first set of experiments, we demonstrate the superiority of our proposed algorithm by evaluating the running time of all the algorithms on all datasets. Note that, the running time of {\distppr} includes both the time of the optimization techniques, i.e., we compare the running time of all algorithms with the same input data. Figure~\ref{fig:exp-efficiency}(a) shows the results of \distppr, \doubling, and {\gxppr} on the datasets in Table~\ref{tbl:exp-real-data}. {\distppr} outperforms all the other methods by up to 2 orders of magnitude on all datasets. This is because {\distppr} addresses the issues of both large and small nodes and is able to leverage the massive parallelism of the distributed computing. However, {\doubling} spends lots of efforts on combining the short walks that results in a huge combinational space, and {\gxppr} incurs significant overhead in the operations of matrix which could be highly sparse. Besides, {\gxppr} is able to process only two small datasets, which again demonstrates that the matrix-based methods are difficult to handle large graphs that require huge memory space. Furthermore, the results of {\doubling} on the datasets, {\it LiveJournal}, {\it Orkut}, {\it Twitter}, and {\it Friendster}, are omitted, since the running time of {\doubling} on those datasets exceeds 24 hours. On the contrary, {\distppr} is able to handle all datasets with a relatively short running time. What's more, compared to {\doubling} and {\gxppr}, the performance of {\distppr} is less sensitive to the number of edges in the graph.

Besides, to evaluate the accuracy of the approximate solutions, we generate the ground-truth by the exact PPR algorithm \cite{wdaj15} that is able to process the graphs with less than ten millions edges on a machine with 48GB memory. Based on that, we evaluate the accuracy of {\distppr} and {\doubling} on 6 datasets by NDCG and MAP, which are widely used to evaluate the results of ranking\footnote{https://spark.apache.org/docs/2.3.0/mllib-evaluation-metrics.html\#ranking-systems}. In particular, given the graph $G=(V,E,w)$, for each node $v \in V$, we rank the nodes $u \in V$ by the value $\pi(v,u)$. We take the first $1000$ nodes from each ranking for evaluation. Figure~\ref{fig:exp-accuracy} shows the results of NDCG and MAP on several datasets for both {\distppr} and {\doubling} respectively. On all datasets, the NDCG of {\distppr} outperforms the one of {\doubling} by up to $21\%$, and the MAP of {\distppr} is also better than the one of {\doubling} by up to $19\%$. This is because {\doubling} does not guarantee the accuracy of results, whereas {\distppr} has a strong accuracy guarantee according to Definition~\ref{def:approx-ppr}.

\subsection{Effects of Optimizations}

In this section of experiments, we study the effects of the three proposed optimization techniques: the alias tree (AT) method for handling large nodes, the big move (BM) technique for accelerating the walks on small nodes, and the technique that auto-tunes $\gamma$.

We first evaluate the techniques of AT and BM. We consider two versions of algorithms modified from {\distppr}: {\distppr} but with the two optimizations disabled (denoted as {\it NA}), and {\distppr} but with only the AT enabled (denoted as {\it AT}). We define the {\it relative overhead} of each modified version of {\distppr} on a dataset $D$ as its running time of $D$ divided by the running time of {\distppr} with all optimizations enabled. Figure~\ref{fig:exp-efficiency}(b) shows the relative overheads of {\it AT} and {\it NA} on each dataset. Note that, the results on {\it Twitter} are not shown, as both {\it AT} and {\it NA} fail to process {\it Twitter} whose node degree is highly skew. The relative overheads of {\it AT} are around $1$ in all cases, which indicates that the AT technique allows the algorithm avoiding the issues of large nodes and brings very few overheads in the running time. Meanwhile, the relative overhead of {\it NA} is around $1.5$ on all datasets, implying that the BM technique improves the efficiency of {\distppr} by a factor of $1.5$.

\begin{figure*}[t]
\centering
\begin{tabular}{cc}
\hspace{-6mm} \includegraphics[height=41mm]{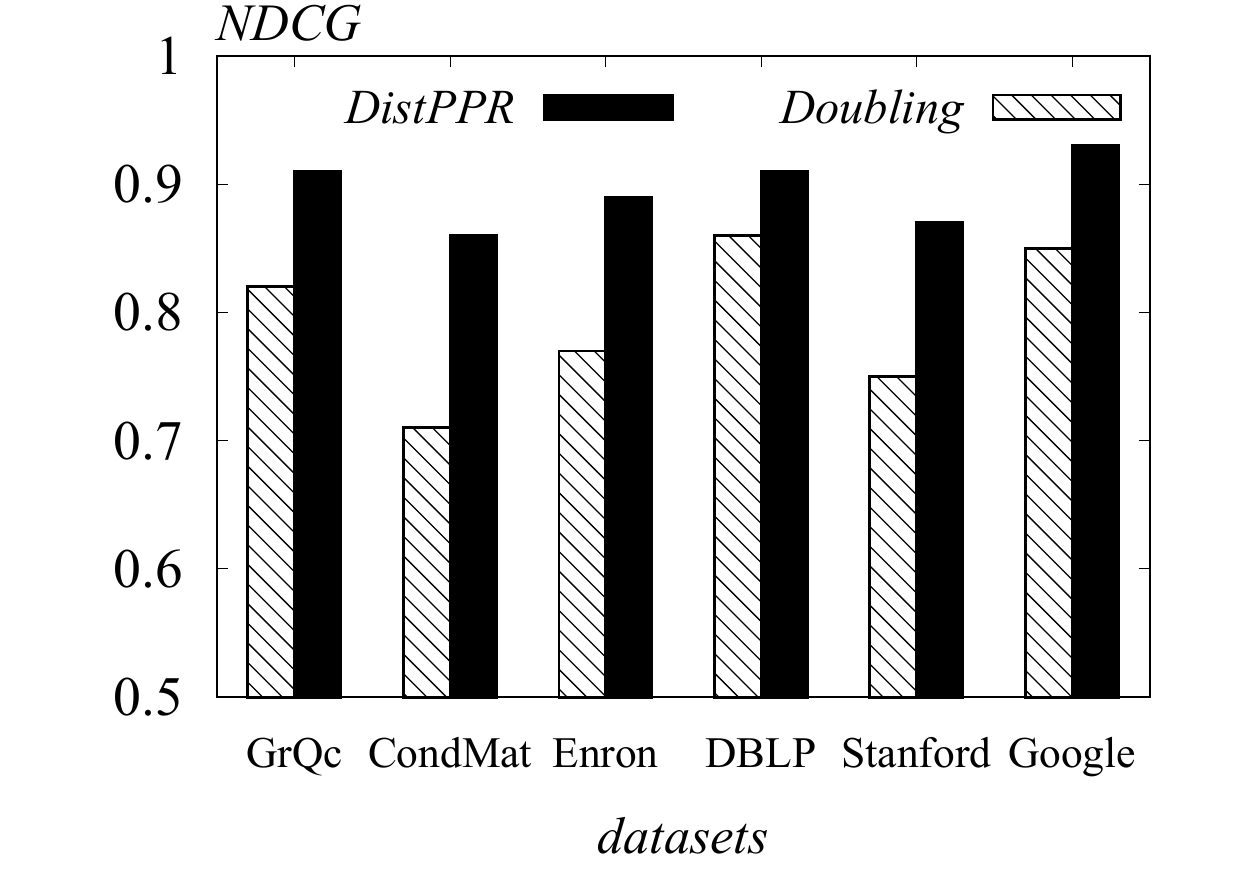}
& \hspace{-2mm} \includegraphics[height=41mm]{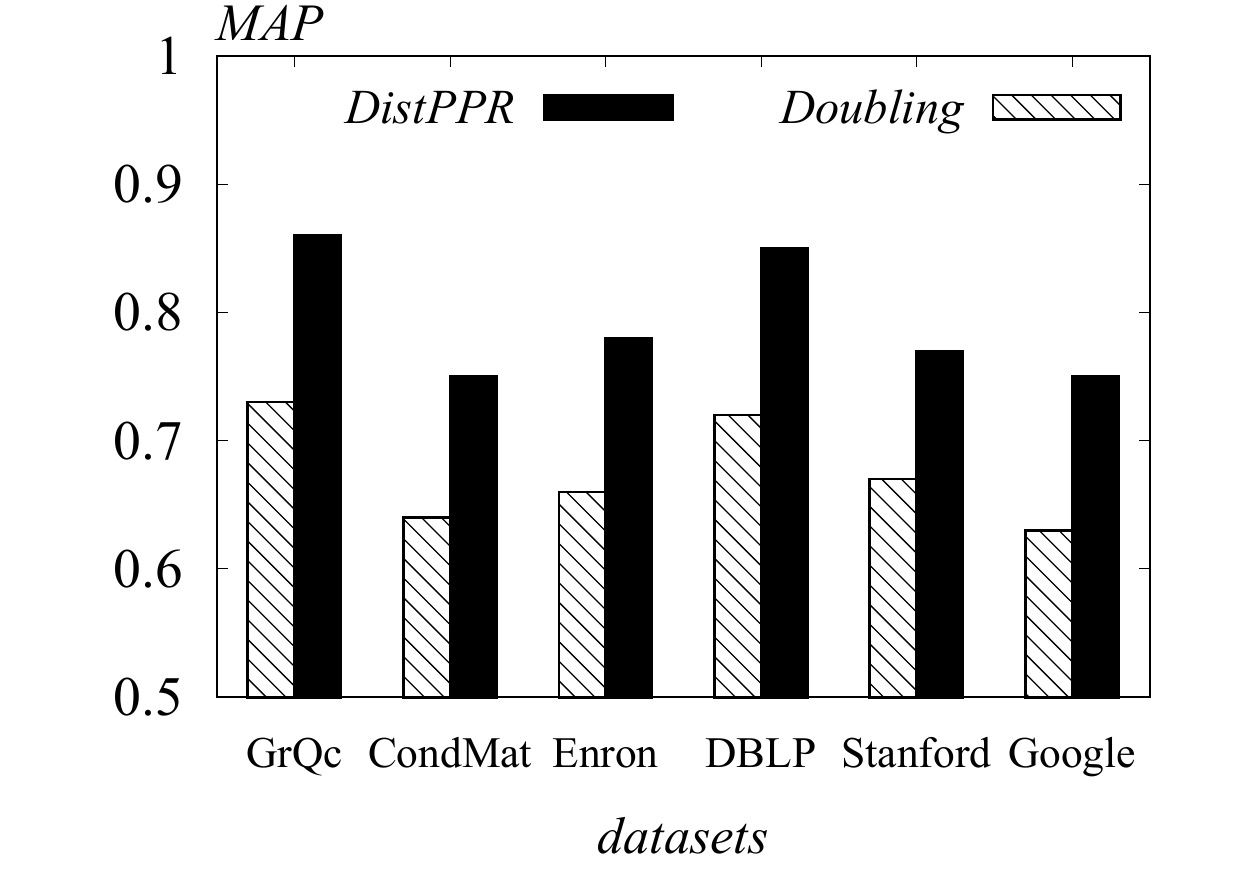}  \\
(a) NDCG. & (b) MAP. \\
\end{tabular}
\vspace{-3mm}
\caption{Accuracy comparison on NDCG and MAP.}
\label{fig:exp-accuracy}
\vspace{-1mm}
\end{figure*}

\begin{figure*}[t!]
\centering
\begin{tabular}{ccc}
\hspace{-10mm} \includegraphics[height=43mm]{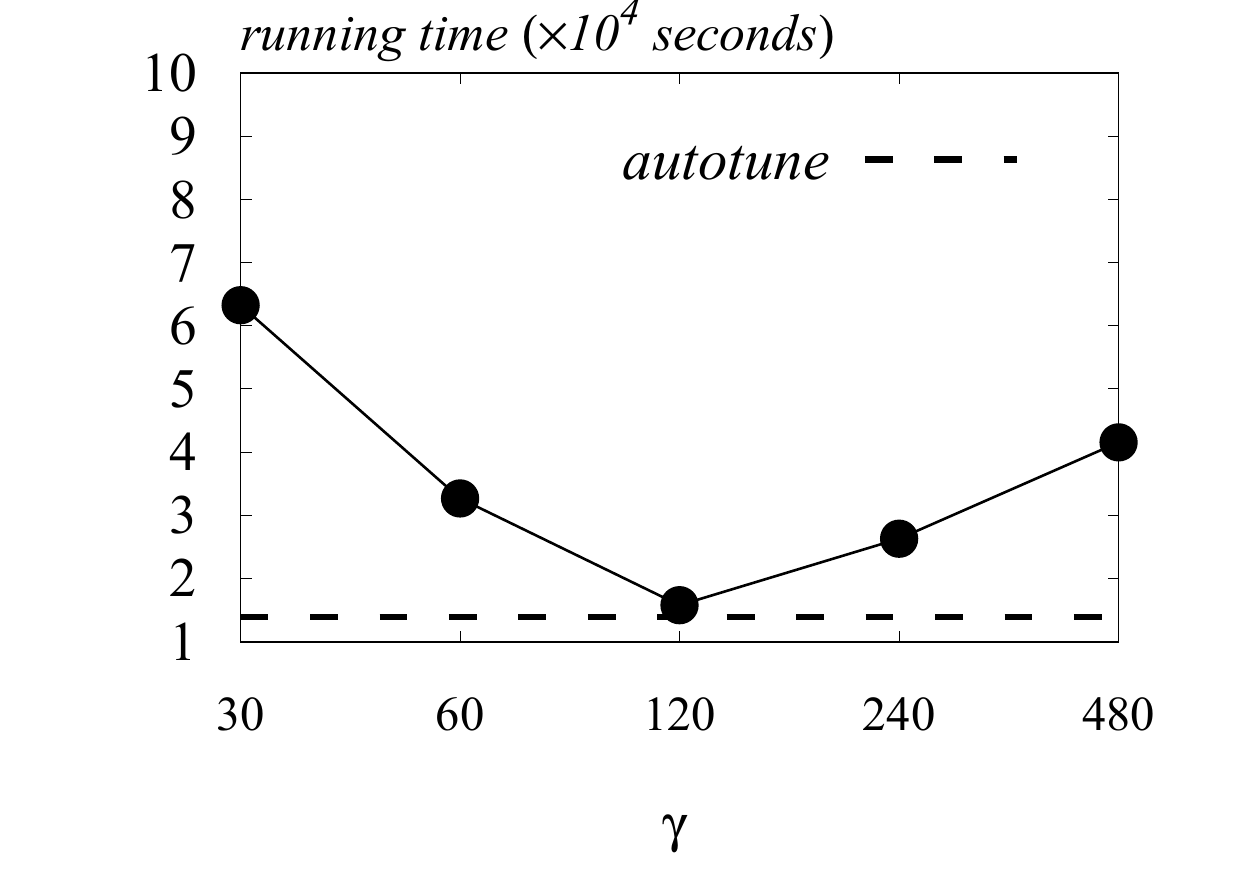}
& \hspace{-10mm} \includegraphics[height=43mm]{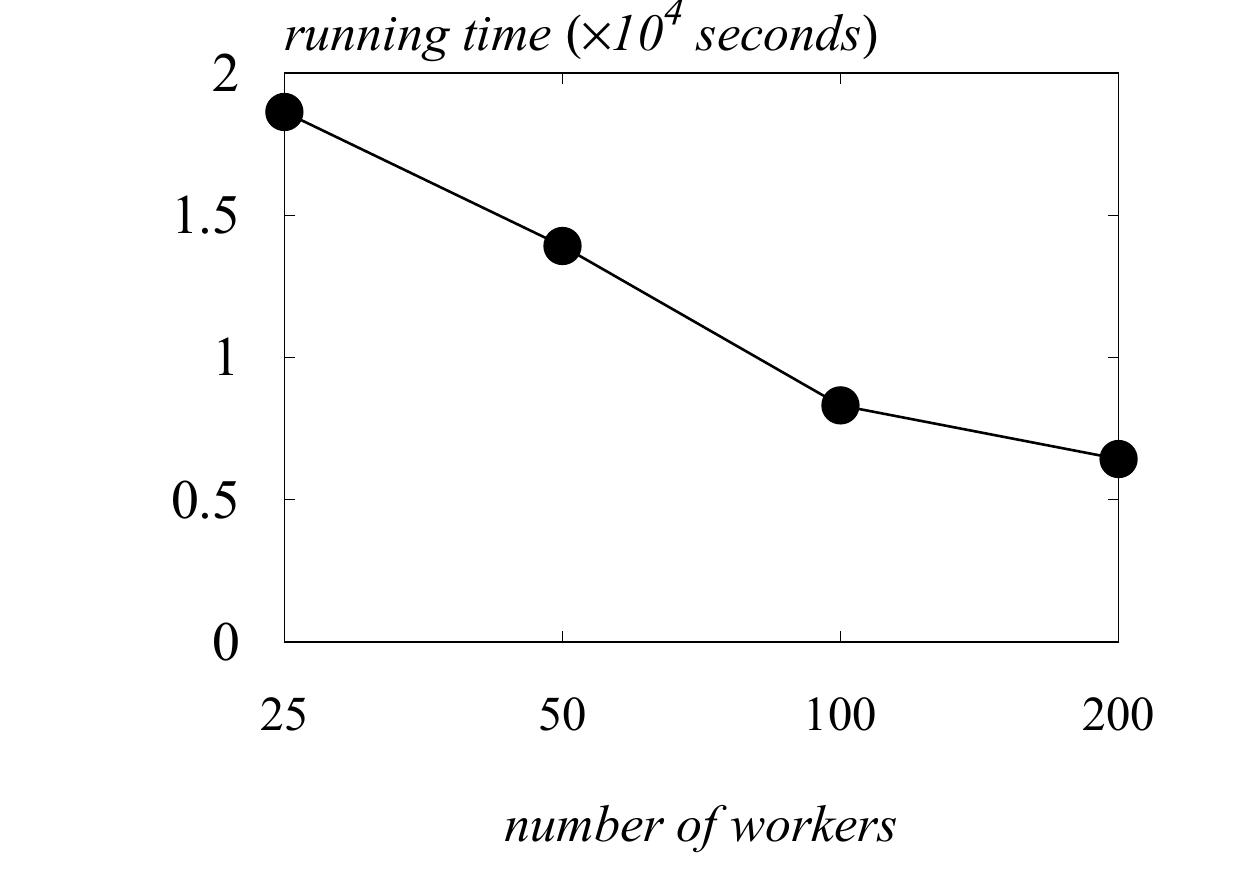}
& \hspace{-10mm} \includegraphics[height=43mm]{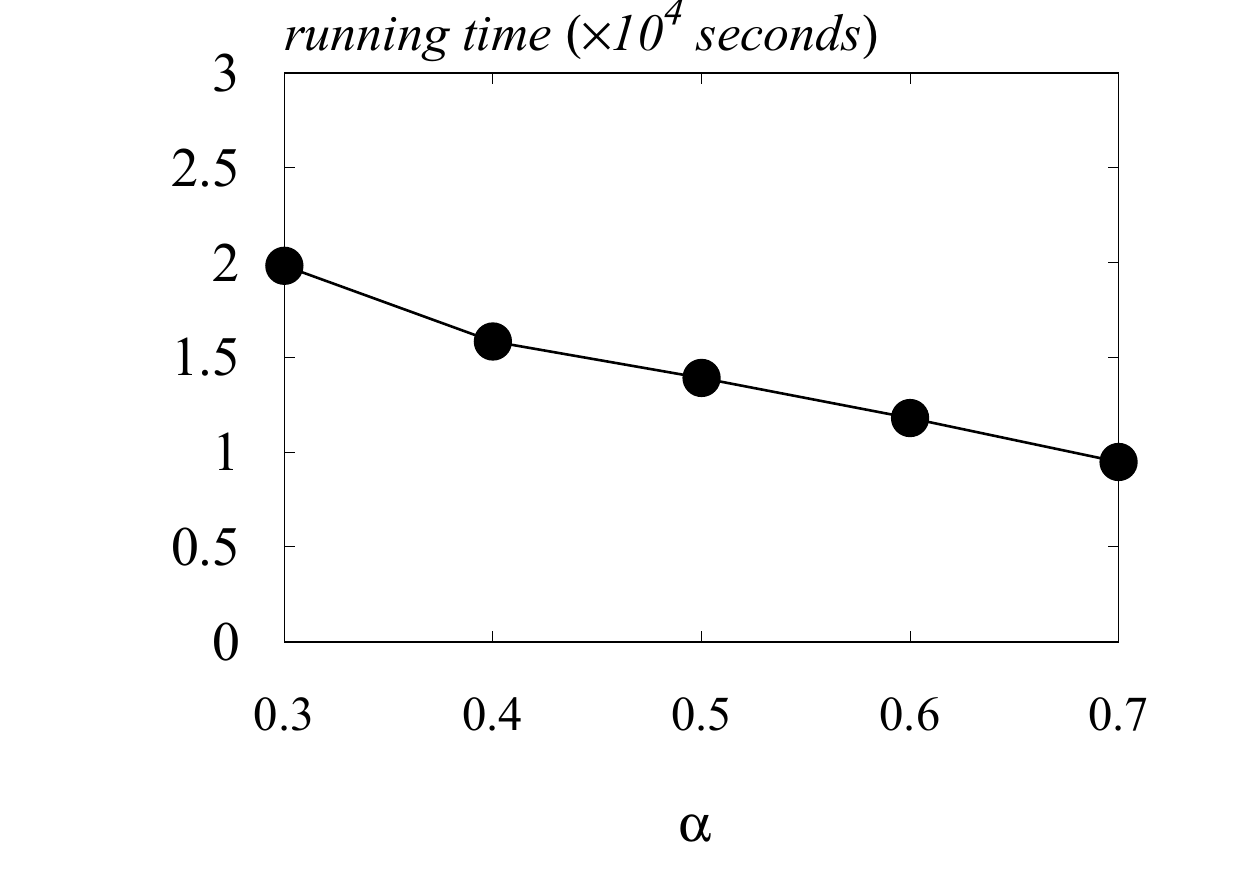}  \\
 (a) Varying $\gamma$. & (b) Varying the number of workers. &  (c) Varying $\alpha$.
\end{tabular}
\vspace{-2mm}
\caption{Varying parameters.}
\label{fig:exp-parameters}
\vspace{-3mm}
\end{figure*}

Then, we evaluate the performance of the technique that automatically tunes $\gamma$. Recall that $\gamma \leq \omega$ is the parameter in the parallel pipeline framework that decides the number of random walks to generate in a pipeline. To evaluate our choice of $\gamma$, we measure the running time of {\distppr} on the largest dataset {\it Friendster} by varying $\gamma$ from $30$ to $480$, and compare with our autotune technique. We plot the results in Figure~\ref{fig:exp-parameters}(a). Observe that the optimal value for $\gamma$ is around $120$, and our automatic choice of $\gamma$ leads to a performance identical to the optimum. On the contrary, compared to the optimum, when $\gamma \leq 120$, it leads to longer running time, since the number of pipelines increases. Note that, the number of pipelines is $\lceil \omega / \gamma \rceil$. On the other hand, when $\gamma \geq 240$, it overloads the cluster with an excessive number of random walks to process simultaneously, which degrades the performance of {\distppr}.

In summary, the three optimization techniques improve the efficiency of {\distppr} by up to 2 times, and help scale {\distppr} to large graphs whose node degrees could be highly skew.

\subsection{Varying Other Parameters}

In the last set of experiments, we study the sensitive of {\distppr} to the other parameters, i.e., the number of workers on Spark and the termination probability $\alpha$. Note that, we omit the studies on the threshold $\delta$, as $\delta$ has the same effect as $\alpha$ to the number $\omega$ of random walks. We adopt the largest dataset {\it Friendster} in all subsequent experiments.

Figure~\ref{fig:exp-parameters}(b) reports the running time of {\distppr} on {\it Friendster} dataset by varying the number of workers on Spark. When we increase the number of workers, the running time of {\distppr} decreases accordingly. However, when the number of workers ranges from $100$ to $200$, the decrease rate of {\distppr}'s running time becomes low, since the cost of communication between workers increases.

Figure~\ref{fig:exp-parameters}(c) shows the performance of {\distppr} by varying the termination probability $\alpha$. When $\alpha$ increases, it increases the probability of a random walk to be terminated at a certain step, leading to the decrease of the total amount of computation.

\section{Related Work} \label{sec:related}

There exist some work \cite{bkd11,aage15} that exploit the distributed computing techniques to accelerate the computation of the fully PPR or PR on large graphs based on the Monte Carlo approximation. They first generate some random walks of relatively short length, and then combine the short walks to obtain the random walks of long length. In particular, given two short walks $P_1 = \langle v_{h1} \cdots v_{t1} \rangle$ and $P_2 = \langle v_{h2} \cdots v_{t2} \rangle$, we can combine $P_1$ and $P_2$ if and only if the tail $v_{t1}$ of $P_1$ equals the head $v_{h2}$ of $P_2$, i.e., $v_{t1} = v_{h2}$, which results in a path $P = \langle v_{h1}, \cdots, v_{t1}, \cdots, v_{t2} \rangle$ of length $\ell(P_1) + \ell(P_2) - 1$. To make the combined paths sufficiently random, however, the number of short walks could be extremely large, especially for the large-degree nodes, rendering these approaches inefficient. Besides, there is no guarantee on the accuracy of the results produced by these approaches \cite{bkd11,aage15}.

Besides, some distributed algorithms are proposed to answer the queries of single-pair PPR \cite{qzjj16} or single-source PPR \cite{qzjj16,txgj+17,jnlu17}, which differ from the problem of this paper that focuses on the computation of fully PPR. Additionally, there are some distributed and parallel algorithms for graph data \cite{lxg14,lxxl17,hlsk+16}, but they do not aim for the problem of fully PPR computation.

On the other hand, Fogaras et al. \cite{dbkt05} developed an approach on a single machine for the fully PPR by the Monte Carlo approximation. In particular, to facilitate the traversal on the graph, their approach relies on an indexing structure of the graph that is required by the computation of random walks for all nodes, which renders this approach difficult to be implemented as a distributed algorithm. Moreover, there also exist a plethora of techniques for processing the queries of single-pair PPR \cite{psas14,syxy+16,psa16}, single-source PPR \cite{jkm13,hcj06,srxz+17}, and top-$k$ PPR \cite{zxxs+18} on a single machine, all of which have a different problem setting against the fully PPR. Besides, these techniques often require a preprocessing step, which are sequential algorithms, rendering them difficult to be translated in parallel.

Finally, Xie et al. \cite{wdaj15} addressed an important issue that edge-weighted graphs widely exist in various applications, which is ignored by previous studies. However, their approach requires to compute based on the adjacency matrix of the graph, which could be prohibitively large.

\section{Conclusions} \label{sec:conclusions}

This paper studies the problem of the fully Personalized PageRank, which has enormous applications in link prediction and recommendation systems for social networks. We approach this problem by devising an efficient distributed algorithm to accelerate the Monte Carlo approximation, which requires to generate a large number of random walks for each node of the graph. We exploit the parallel pipeline framework to achieve the superiority of cluster parallelism. To optimize the proposed algorithm, we develop three techniques that significantly improve the performance of the proposed algorithm in terms of efficiency and scalability. With extensive experiments on various real-life networks, we demonstrate that our proposed solution is up to $2$ orders of magnitude faster than the start-of-the-art solutions, and also outperforms the baseline solutions on the evaluation of accuracy.

\begin{balance}
\bibliographystyle{ACM-Reference-Format}
\bibliography{bigppr}


\begin{thebibliography}{31}


\ifx \showCODEN    \undefined \def \showCODEN     #1{\unskip}     \fi
\ifx \showDOI      \undefined \def \showDOI       #1{#1}\fi
\ifx \showISBNx    \undefined \def \showISBNx     #1{\unskip}     \fi
\ifx \showISBNxiii \undefined \def \showISBNxiii  #1{\unskip}     \fi
\ifx \showISSN     \undefined \def \showISSN      #1{\unskip}     \fi
\ifx \showLCCN     \undefined \def \showLCCN      #1{\unskip}     \fi
\ifx \shownote     \undefined \def \shownote      #1{#1}          \fi
\ifx \showarticletitle \undefined \def \showarticletitle #1{#1}   \fi
\ifx \showURL      \undefined \def \showURL       {\relax}        \fi
\providecommand\bibfield[2]{#2}
\providecommand\bibinfo[2]{#2}
\providecommand\natexlab[1]{#1}
\providecommand\showeprint[2][]{arXiv:#2}

\bibitem[\protect\citeauthoryear{Agirre and Soroa}{Agirre and Soroa}{2009}]%
        {ea09}
\bibfield{author}{\bibinfo{person}{Eneko Agirre} {and} \bibinfo{person}{Aitor
  Soroa}.} \bibinfo{year}{2009}\natexlab{}.
\newblock \showarticletitle{Personalizing PageRank for Word Sense
  Disambiguation}. In \bibinfo{booktitle}{\emph{{EACL} 2009, 12th Conference of
  the European Chapter of the Association for Computational Linguistics,
  Proceedings of the Conference, Athens, Greece, March 30 - April 3, 2009}}.
  \bibinfo{pages}{33--41}.
\newblock


\bibitem[\protect\citeauthoryear{Bahmani, Chakrabarti, and Xin}{Bahmani
  et~al\mbox{.}}{2011}]%
        {bkd11}
\bibfield{author}{\bibinfo{person}{Bahman Bahmani}, \bibinfo{person}{Kaushik
  Chakrabarti}, {and} \bibinfo{person}{Dong Xin}.}
  \bibinfo{year}{2011}\natexlab{}.
\newblock \showarticletitle{Fast personalized PageRank on MapReduce}. In
  \bibinfo{booktitle}{\emph{Proceedings of the {ACM} {SIGMOD} International
  Conference on Management of Data, {SIGMOD} 2011, Athens, Greece, June 12-16,
  2011}}. \bibinfo{pages}{973--984}.
\newblock


\bibitem[\protect\citeauthoryear{Dean and Ghemawat}{Dean and Ghemawat}{2004}]%
        {js04}
\bibfield{author}{\bibinfo{person}{Jeffrey Dean} {and} \bibinfo{person}{Sanjay
  Ghemawat}.} \bibinfo{year}{2004}\natexlab{}.
\newblock \showarticletitle{MapReduce: Simplified Data Processing on Large
  Clusters}. In \bibinfo{booktitle}{\emph{6th Symposium on Operating System
  Design and Implementation {(OSDI} 2004), San Francisco, California, USA,
  December 6-8, 2004}}. \bibinfo{pages}{137--150}.
\newblock


\bibitem[\protect\citeauthoryear{Efraimidis and Spirakis}{Efraimidis and
  Spirakis}{2006}]%
        {pp06}
\bibfield{author}{\bibinfo{person}{Pavlos~S. Efraimidis} {and}
  \bibinfo{person}{Paul~G. Spirakis}.} \bibinfo{year}{2006}\natexlab{}.
\newblock \showarticletitle{Weighted random sampling with a reservoir}.
\newblock \bibinfo{journal}{\emph{Inf. Process. Lett.}} \bibinfo{volume}{97},
  \bibinfo{number}{5} (\bibinfo{year}{2006}), \bibinfo{pages}{181--185}.
\newblock


\bibitem[\protect\citeauthoryear{Eksombatchai, Jindal, Liu, Liu, Sharma,
  Sugnet, Ulrich, and Leskovec}{Eksombatchai et~al\mbox{.}}{2018}]%
        {cpjy+18}
\bibfield{author}{\bibinfo{person}{Chantat Eksombatchai},
  \bibinfo{person}{Pranav Jindal}, \bibinfo{person}{Jerry~Zitao Liu},
  \bibinfo{person}{Yuchen Liu}, \bibinfo{person}{Rahul Sharma},
  \bibinfo{person}{Charles Sugnet}, \bibinfo{person}{Mark Ulrich}, {and}
  \bibinfo{person}{Jure Leskovec}.} \bibinfo{year}{2018}\natexlab{}.
\newblock \showarticletitle{Pixie: {A} System for Recommending 3+ Billion Items
  to 200+ Million Users in Real-Time}. In \bibinfo{booktitle}{\emph{Proceedings
  of the 2018 World Wide Web Conference on World Wide Web, {WWW} 2018, Lyon,
  France, April 23-27, 2018}}. \bibinfo{pages}{1775--1784}.
\newblock


\bibitem[\protect\citeauthoryear{Faloutsos, Faloutsos, and Faloutsos}{Faloutsos
  et~al\mbox{.}}{1999}]%
        {mpc99}
\bibfield{author}{\bibinfo{person}{Michalis Faloutsos}, \bibinfo{person}{Petros
  Faloutsos}, {and} \bibinfo{person}{Christos Faloutsos}.}
  \bibinfo{year}{1999}\natexlab{}.
\newblock \showarticletitle{On Power-law Relationships of the Internet
  Topology}. In \bibinfo{booktitle}{\emph{{SIGCOMM}}}.
  \bibinfo{pages}{251--262}.
\newblock


\bibitem[\protect\citeauthoryear{Fogaras, R{\'{a}}cz, Csalog{\'{a}}ny, and
  Sarl{\'{o}}s}{Fogaras et~al\mbox{.}}{2005}]%
        {dbkt05}
\bibfield{author}{\bibinfo{person}{D{\'{a}}niel Fogaras},
  \bibinfo{person}{Bal{\'{a}}zs R{\'{a}}cz}, \bibinfo{person}{K{\'{a}}roly
  Csalog{\'{a}}ny}, {and} \bibinfo{person}{Tam{\'{a}}s Sarl{\'{o}}s}.}
  \bibinfo{year}{2005}\natexlab{}.
\newblock \showarticletitle{Towards Scaling Fully Personalized PageRank:
  Algorithms, Lower Bounds, and Experiments}.
\newblock \bibinfo{journal}{\emph{Internet Mathematics}} \bibinfo{volume}{2},
  \bibinfo{number}{3} (\bibinfo{year}{2005}), \bibinfo{pages}{333--358}.
\newblock


\bibitem[\protect\citeauthoryear{Fujiwara, Nakatsuji, Yamamuro, Shiokawa, and
  Onizuka}{Fujiwara et~al\mbox{.}}{2012}]%
        {ymth+12}
\bibfield{author}{\bibinfo{person}{Yasuhiro Fujiwara}, \bibinfo{person}{Makoto
  Nakatsuji}, \bibinfo{person}{Takeshi Yamamuro}, \bibinfo{person}{Hiroaki
  Shiokawa}, {and} \bibinfo{person}{Makoto Onizuka}.}
  \bibinfo{year}{2012}\natexlab{}.
\newblock \showarticletitle{Efficient personalized pagerank with accuracy
  assurance}. In \bibinfo{booktitle}{\emph{The 18th {ACM} {SIGKDD}
  International Conference on Knowledge Discovery and Data Mining, {KDD} '12,
  Beijing, China, August 12-16, 2012}}. \bibinfo{pages}{15--23}.
\newblock


\bibitem[\protect\citeauthoryear{Guo, Cao, Cong, Lu, and Lin}{Guo
  et~al\mbox{.}}{2017}]%
        {txgj+17}
\bibfield{author}{\bibinfo{person}{Tao Guo}, \bibinfo{person}{Xin Cao},
  \bibinfo{person}{Gao Cong}, \bibinfo{person}{Jiaheng Lu}, {and}
  \bibinfo{person}{Xuemin Lin}.} \bibinfo{year}{2017}\natexlab{}.
\newblock \showarticletitle{Distributed Algorithms on Exact Personalized
  PageRank}. In \bibinfo{booktitle}{\emph{Proceedings of the 2017 {ACM}
  International Conference on Management of Data, {SIGMOD} Conference 2017,
  Chicago, IL, USA, May 14-19, 2017}}. \bibinfo{pages}{479--494}.
\newblock


\bibitem[\protect\citeauthoryear{Ho, Lin, Shaham, Krishnaswamy, Dang, Wang,
  Zhongyan, and She{-}Nash}{Ho et~al\mbox{.}}{2016}]%
        {hlsk+16}
\bibfield{author}{\bibinfo{person}{Qirong Ho}, \bibinfo{person}{Wenqing Lin},
  \bibinfo{person}{Eran Shaham}, \bibinfo{person}{Shonali Krishnaswamy},
  \bibinfo{person}{The~Anh Dang}, \bibinfo{person}{Jingxuan Wang},
  \bibinfo{person}{Isabel~Choo Zhongyan}, {and} \bibinfo{person}{Amy
  She{-}Nash}.} \bibinfo{year}{2016}\natexlab{}.
\newblock \showarticletitle{A Distributed Graph Algorithm for Discovering
  Unique Behavioral Groups from Large-Scale Telco Data}. In
  \bibinfo{booktitle}{\emph{Proceedings of the 25th {ACM} International
  Conference on Information and Knowledge Management, {CIKM} 2016,
  Indianapolis, IN, USA, October 24-28, 2016}}. \bibinfo{pages}{1353--1362}.
\newblock


\bibitem[\protect\citeauthoryear{Jung, Park, Sael, and Kang}{Jung
  et~al\mbox{.}}{2017}]%
        {jnlu17}
\bibfield{author}{\bibinfo{person}{Jinhong Jung}, \bibinfo{person}{Namyong
  Park}, \bibinfo{person}{Lee Sael}, {and} \bibinfo{person}{U. Kang}.}
  \bibinfo{year}{2017}\natexlab{}.
\newblock \showarticletitle{BePI: Fast and Memory-Efficient Method for
  Billion-Scale Random Walk with Restart}. In
  \bibinfo{booktitle}{\emph{Proceedings of the 2017 {ACM} International
  Conference on Management of Data, {SIGMOD} Conference 2017, Chicago, IL, USA,
  May 14-19, 2017}}. \bibinfo{pages}{789--804}.
\newblock


\bibitem[\protect\citeauthoryear{Kim, Candan, and Sapino}{Kim
  et~al\mbox{.}}{2013}]%
        {jkm13}
\bibfield{author}{\bibinfo{person}{Jung~Hyun Kim},
  \bibinfo{person}{K.~Sel{\c{c}}uk Candan}, {and} \bibinfo{person}{Maria~Luisa
  Sapino}.} \bibinfo{year}{2013}\natexlab{}.
\newblock \showarticletitle{{LR-PPR:} locality-sensitive, re-use promoting,
  approximate personalized pagerank computation}. In
  \bibinfo{booktitle}{\emph{22nd {ACM} International Conference on Information
  and Knowledge Management, CIKM'13, San Francisco, CA, USA, October 27 -
  November 1, 2013}}. \bibinfo{pages}{1801--1806}.
\newblock


\bibitem[\protect\citeauthoryear{Liao, Choudhary, Weiner, and Varshney}{Liao
  et~al\mbox{.}}{2005}]%
        {wadp05}
\bibfield{author}{\bibinfo{person}{Wei{-}keng Liao}, \bibinfo{person}{Alok~N.
  Choudhary}, \bibinfo{person}{Donald~D. Weiner}, {and}
  \bibinfo{person}{Pramod~K. Varshney}.} \bibinfo{year}{2005}\natexlab{}.
\newblock \showarticletitle{Performance Evaluation of a Parallel Pipeline
  Computational Model for Space-Time Adaptive Processing}.
\newblock \bibinfo{journal}{\emph{The Journal of Supercomputing}}
  \bibinfo{volume}{31}, \bibinfo{number}{2} (\bibinfo{year}{2005}),
  \bibinfo{pages}{137--160}.
\newblock


\bibitem[\protect\citeauthoryear{Lin, Xiao, and Ghinita}{Lin
  et~al\mbox{.}}{2014}]%
        {lxg14}
\bibfield{author}{\bibinfo{person}{Wenqing Lin}, \bibinfo{person}{Xiaokui
  Xiao}, {and} \bibinfo{person}{Gabriel Ghinita}.}
  \bibinfo{year}{2014}\natexlab{}.
\newblock \showarticletitle{Large-scale frequent subgraph mining in MapReduce}.
  In \bibinfo{booktitle}{\emph{{IEEE} 30th International Conference on Data
  Engineering, Chicago, {ICDE} 2014, IL, USA, March 31 - April 4, 2014}}.
  \bibinfo{pages}{844--855}.
\newblock


\bibitem[\protect\citeauthoryear{Lin, Xiao, Xie, and Li}{Lin
  et~al\mbox{.}}{2017}]%
        {lxxl17}
\bibfield{author}{\bibinfo{person}{Wenqing Lin}, \bibinfo{person}{Xiaokui
  Xiao}, \bibinfo{person}{Xing Xie}, {and} \bibinfo{person}{Xiaoli Li}.}
  \bibinfo{year}{2017}\natexlab{}.
\newblock \showarticletitle{Network Motif Discovery: {A} {GPU} Approach}.
\newblock \bibinfo{journal}{\emph{{IEEE} Trans. Knowl. Data Eng.}}
  \bibinfo{volume}{29}, \bibinfo{number}{3} (\bibinfo{year}{2017}),
  \bibinfo{pages}{513--528}.
\newblock


\bibitem[\protect\citeauthoryear{Liu, Li, Lui, and Cheng}{Liu
  et~al\mbox{.}}{2016}]%
        {qzjj16}
\bibfield{author}{\bibinfo{person}{Qin Liu}, \bibinfo{person}{Zhenguo Li},
  \bibinfo{person}{John C.~S. Lui}, {and} \bibinfo{person}{Jiefeng Cheng}.}
  \bibinfo{year}{2016}\natexlab{}.
\newblock \showarticletitle{PowerWalk: Scalable Personalized PageRank via
  Random Walks with Vertex-Centric Decomposition}. In
  \bibinfo{booktitle}{\emph{Proceedings of the 25th {ACM} International
  Conference on Information and Knowledge Management, {CIKM} 2016,
  Indianapolis, IN, USA, October 24-28, 2016}}. \bibinfo{pages}{195--204}.
\newblock


\bibitem[\protect\citeauthoryear{Lofgren, Banerjee, and Goel}{Lofgren
  et~al\mbox{.}}{2016}]%
        {psa16}
\bibfield{author}{\bibinfo{person}{Peter Lofgren}, \bibinfo{person}{Siddhartha
  Banerjee}, {and} \bibinfo{person}{Ashish Goel}.}
  \bibinfo{year}{2016}\natexlab{}.
\newblock \showarticletitle{Personalized PageRank Estimation and Search: {A}
  Bidirectional Approach}. In \bibinfo{booktitle}{\emph{Proceedings of the
  Ninth {ACM} International Conference on Web Search and Data Mining, San
  Francisco, CA, USA, February 22-25, 2016}}. \bibinfo{pages}{163--172}.
\newblock


\bibitem[\protect\citeauthoryear{Lofgren, Banerjee, Goel, and Comandur}{Lofgren
  et~al\mbox{.}}{2014}]%
        {psas14}
\bibfield{author}{\bibinfo{person}{Peter Lofgren}, \bibinfo{person}{Siddhartha
  Banerjee}, \bibinfo{person}{Ashish Goel}, {and} \bibinfo{person}{Seshadhri
  Comandur}.} \bibinfo{year}{2014}\natexlab{}.
\newblock \showarticletitle{{FAST-PPR:} scaling personalized pagerank
  estimation for large graphs}. In \bibinfo{booktitle}{\emph{The 20th {ACM}
  {SIGKDD} International Conference on Knowledge Discovery and Data Mining,
  {KDD} '14, New York, NY, {USA} - August 24 - 27, 2014}}.
  \bibinfo{pages}{1436--1445}.
\newblock


\bibitem[\protect\citeauthoryear{Luo, Xiao, Lin, and Kao}{Luo
  et~al\mbox{.}}{2019}]%
        {lxlk19}
\bibfield{author}{\bibinfo{person}{Siqiang Luo}, \bibinfo{person}{Xiaokui
  Xiao}, \bibinfo{person}{Wenqing Lin}, {and} \bibinfo{person}{Ben Kao}.}
  \bibinfo{year}{2019}\natexlab{}.
\newblock \showarticletitle{Efficient Batch One-Hop Personalized PageRanks}. In
  \bibinfo{booktitle}{\emph{{IEEE} 35th International Conference on Data
  Engineering, Macau SAR, {ICDE} 2019, China, April 8 - April 12, 2019}}.
\newblock


\bibitem[\protect\citeauthoryear{Malewicz, Austern, Bik, Dehnert, Horn, Leiser,
  and Czajkowski}{Malewicz et~al\mbox{.}}{2010}]%
        {gmaj+10}
\bibfield{author}{\bibinfo{person}{Grzegorz Malewicz},
  \bibinfo{person}{Matthew~H. Austern}, \bibinfo{person}{Aart J.~C. Bik},
  \bibinfo{person}{James~C. Dehnert}, \bibinfo{person}{Ilan Horn},
  \bibinfo{person}{Naty Leiser}, {and} \bibinfo{person}{Grzegorz Czajkowski}.}
  \bibinfo{year}{2010}\natexlab{}.
\newblock \showarticletitle{Pregel: a system for large-scale graph processing}.
  In \bibinfo{booktitle}{\emph{Proceedings of the {ACM} {SIGMOD} International
  Conference on Management of Data, {SIGMOD} 2010, Indianapolis, Indiana, USA,
  June 6-10, 2010}}. \bibinfo{pages}{135--146}.
\newblock


\bibitem[\protect\citeauthoryear{Nguyen, Tomeo, Noia, and Sciascio}{Nguyen
  et~al\mbox{.}}{2015}]%
        {ppte15}
\bibfield{author}{\bibinfo{person}{Phuong Nguyen}, \bibinfo{person}{Paolo
  Tomeo}, \bibinfo{person}{Tommaso~Di Noia}, {and} \bibinfo{person}{Eugenio~Di
  Sciascio}.} \bibinfo{year}{2015}\natexlab{}.
\newblock \showarticletitle{An evaluation of SimRank and Personalized PageRank
  to build a recommender system for the Web of Data}. In
  \bibinfo{booktitle}{\emph{Proceedings of the 24th International Conference on
  World Wide Web Companion, {WWW} 2015, Florence, Italy, May 18-22, 2015 -
  Companion Volume}}. \bibinfo{pages}{1477--1482}.
\newblock


\bibitem[\protect\citeauthoryear{Sarma, Molla, Pandurangan, and Upfal}{Sarma
  et~al\mbox{.}}{2015}]%
        {aage15}
\bibfield{author}{\bibinfo{person}{Atish~Das Sarma},
  \bibinfo{person}{Anisur~Rahaman Molla}, \bibinfo{person}{Gopal Pandurangan},
  {and} \bibinfo{person}{Eli Upfal}.} \bibinfo{year}{2015}\natexlab{}.
\newblock \showarticletitle{Fast distributed PageRank computation}.
\newblock \bibinfo{journal}{\emph{Theor. Comput. Sci.}}  \bibinfo{volume}{561}
  (\bibinfo{year}{2015}), \bibinfo{pages}{113--121}.
\newblock


\bibitem[\protect\citeauthoryear{Tao, Lin, and Xiao}{Tao et~al\mbox{.}}{2013}]%
        {ywx13}
\bibfield{author}{\bibinfo{person}{Yufei Tao}, \bibinfo{person}{Wenqing Lin},
  {and} \bibinfo{person}{Xiaokui Xiao}.} \bibinfo{year}{2013}\natexlab{}.
\newblock \showarticletitle{Minimal MapReduce algorithms}. In
  \bibinfo{booktitle}{\emph{Proceedings of the {ACM} {SIGMOD} International
  Conference on Management of Data, {SIGMOD} 2013, New York, NY, USA, June
  22-27, 2013}}. \bibinfo{pages}{529--540}.
\newblock


\bibitem[\protect\citeauthoryear{Tong, Faloutsos, and Pan}{Tong
  et~al\mbox{.}}{2006}]%
        {hcj06}
\bibfield{author}{\bibinfo{person}{Hanghang Tong}, \bibinfo{person}{Christos
  Faloutsos}, {and} \bibinfo{person}{Jia{-}Yu Pan}.}
  \bibinfo{year}{2006}\natexlab{}.
\newblock \showarticletitle{Fast Random Walk with Restart and Its
  Applications}. In \bibinfo{booktitle}{\emph{Proceedings of the 6th {IEEE}
  International Conference on Data Mining {(ICDM} 2006), 18-22 December 2006,
  Hong Kong, China}}. \bibinfo{pages}{613--622}.
\newblock


\bibitem[\protect\citeauthoryear{Vose}{Vose}{1991}]%
        {m91}
\bibfield{author}{\bibinfo{person}{Michael~D. Vose}.}
  \bibinfo{year}{1991}\natexlab{}.
\newblock \showarticletitle{A Linear Algorithm For Generating Random Numbers
  With a Given Distribution}.
\newblock \bibinfo{journal}{\emph{{IEEE} Trans. Software Eng.}}
  \bibinfo{volume}{17}, \bibinfo{number}{9} (\bibinfo{year}{1991}),
  \bibinfo{pages}{972--975}.
\newblock


\bibitem[\protect\citeauthoryear{Wang, Tang, Xiao, Yang, and Li}{Wang
  et~al\mbox{.}}{2016}]%
        {syxy+16}
\bibfield{author}{\bibinfo{person}{Sibo Wang}, \bibinfo{person}{Youze Tang},
  \bibinfo{person}{Xiaokui Xiao}, \bibinfo{person}{Yin Yang}, {and}
  \bibinfo{person}{Zengxiang Li}.} \bibinfo{year}{2016}\natexlab{}.
\newblock \showarticletitle{HubPPR: Effective Indexing for Approximate
  Personalized PageRank}.
\newblock \bibinfo{journal}{\emph{{PVLDB}}} \bibinfo{volume}{10},
  \bibinfo{number}{3} (\bibinfo{year}{2016}), \bibinfo{pages}{205--216}.
\newblock


\bibitem[\protect\citeauthoryear{Wang, Yang, Xiao, Wei, and Yang}{Wang
  et~al\mbox{.}}{2017}]%
        {srxz+17}
\bibfield{author}{\bibinfo{person}{Sibo Wang}, \bibinfo{person}{Renchi Yang},
  \bibinfo{person}{Xiaokui Xiao}, \bibinfo{person}{Zhewei Wei}, {and}
  \bibinfo{person}{Yin Yang}.} \bibinfo{year}{2017}\natexlab{}.
\newblock \showarticletitle{{FORA:} Simple and Effective Approximate
  Single-Source Personalized PageRank}. In
  \bibinfo{booktitle}{\emph{Proceedings of the 23rd {ACM} {SIGKDD}
  International Conference on Knowledge Discovery and Data Mining, Halifax, NS,
  Canada, August 13 - 17, 2017}}. \bibinfo{pages}{505--514}.
\newblock


\bibitem[\protect\citeauthoryear{Wei, He, Xiao, Wang, Shang, and Wen}{Wei
  et~al\mbox{.}}{2018}]%
        {zxxs+18}
\bibfield{author}{\bibinfo{person}{Zhewei Wei}, \bibinfo{person}{Xiaodong He},
  \bibinfo{person}{Xiaokui Xiao}, \bibinfo{person}{Sibo Wang},
  \bibinfo{person}{Shuo Shang}, {and} \bibinfo{person}{Ji{-}Rong Wen}.}
  \bibinfo{year}{2018}\natexlab{}.
\newblock \showarticletitle{TopPPR: Top-k Personalized PageRank Queries with
  Precision Guarantees on Large Graphs}. In
  \bibinfo{booktitle}{\emph{Proceedings of the 2018 International Conference on
  Management of Data, {SIGMOD} Conference 2018, Houston, TX, USA, June 10-15,
  2018}}. \bibinfo{pages}{441--456}.
\newblock


\bibitem[\protect\citeauthoryear{Xie, Bindel, Demers, and Gehrke}{Xie
  et~al\mbox{.}}{2015}]%
        {wdaj15}
\bibfield{author}{\bibinfo{person}{Wenlei Xie}, \bibinfo{person}{David Bindel},
  \bibinfo{person}{Alan~J. Demers}, {and} \bibinfo{person}{Johannes Gehrke}.}
  \bibinfo{year}{2015}\natexlab{}.
\newblock \showarticletitle{Edge-Weighted Personalized PageRank: Breaking {A}
  Decade-Old Performance Barrier}. In \bibinfo{booktitle}{\emph{Proceedings of
  the 21th {ACM} {SIGKDD} International Conference on Knowledge Discovery and
  Data Mining, Sydney, NSW, Australia, August 10-13, 2015}}.
  \bibinfo{pages}{1325--1334}.
\newblock


\bibitem[\protect\citeauthoryear{Zaharia, Chowdhury, Franklin, Shenker, and
  Stoica}{Zaharia et~al\mbox{.}}{2010}]%
        {mmms+10}
\bibfield{author}{\bibinfo{person}{Matei Zaharia}, \bibinfo{person}{Mosharaf
  Chowdhury}, \bibinfo{person}{Michael~J. Franklin}, \bibinfo{person}{Scott
  Shenker}, {and} \bibinfo{person}{Ion Stoica}.}
  \bibinfo{year}{2010}\natexlab{}.
\newblock \showarticletitle{Spark: Cluster Computing with Working Sets}. In
  \bibinfo{booktitle}{\emph{2nd {USENIX} Workshop on Hot Topics in Cloud
  Computing, HotCloud'10, Boston, MA, USA, June 22, 2010}}.
\newblock


\bibitem[\protect\citeauthoryear{Zhang, Li, Xu, and Zhang}{Zhang
  et~al\mbox{.}}{2016}]%
        {wdyy16}
\bibfield{author}{\bibinfo{person}{Wanxin Zhang}, \bibinfo{person}{Dongsheng
  Li}, \bibinfo{person}{Ying Xu}, {and} \bibinfo{person}{Yiming Zhang}.}
  \bibinfo{year}{2016}\natexlab{}.
\newblock \showarticletitle{Shuffle-efficient distributed Locality Sensitive
  Hashing on spark}. In \bibinfo{booktitle}{\emph{{IEEE} Conference on Computer
  Communications Workshops, {INFOCOM} Workshops 2016, San Francisco, CA, USA,
  April 10-14, 2016}}. \bibinfo{pages}{766--767}.
\newblock


\end{thebibliography}
\end{balance}

\end{sloppy}
\end{document}